\documentclass[submission,copyright,creativecommons]{eptcs}
\providecommand{\event}{AiML 2026} 

\usepackage{aiml26}

\usepackage{iftex}

\ifpdf
  \usepackage{underscore}         
  \usepackage[T1]{fontenc}        
\else
  \usepackage{breakurl}           
\fi

\usepackage[utf8]{inputenc}
\usepackage{hyperref}
\usepackage{amsthm}
\usepackage{amssymb}
\usepackage{graphicx}
\usepackage{textcomp}
\usepackage{amsmath}
\usepackage{stmaryrd}
\usepackage{url}
\usepackage{color}
\usepackage[all,cmtip]{xy}
\usepackage{caption}
\usepackage{multicol,multienum}

\usepackage{paralist}
\usepackage[shortlabels,inline]{enumitem}
\usepackage{tikz}
\usetikzlibrary{positioning, backgrounds, fit, arrows.meta}
\usepackage{trfsigns}
\usepackage{relsize}
\usepackage{scalerel}
\SelectTips{xy}{12}
\objectmargin={1pt}
\usepackage{proof}

\renewcommand{\iff}{\mathrm{iff}}

\newcommand{\inset}[2]{\left\{\, {#1} \,|\, {#2} \,\right\}}

\newcommand{\setof}[1]{\{\,{#1}\,\}}

\newcommand{\bdia}{\mathop{\blacklozenge}}
\newcommand{\dia}{\mathop{\lozenge}}
\newcommand{\bbox}{\mathop{\blacksquare}}

\newcommand{\Exi}[1]{\exists\, {#1}\,.}
\newcommand{\Any}[1]{\forall\, {#1}\,.}

\newcommand{\bisim}{\leftrightarroweq}
\newcommand{\moeq}{\leftrightsquigarrow}

\renewcommand{\nabla}{\mathlarger{\mathlarger{\triangledown}}}
\newcommand{\bnabla}{\mathlarger{\mathlarger{\blacktriangledown}}}

\newtheorem{claim3}{\textbf{Claim}}
\newenvironment{clma}{\begin{claim3}\rm}{\end{claim3}\rm}
\newenvironment{clma1}{\setcounter{claim3}{0}
              \begin{claim3}\rm}{\end{claim3}\rm}
\newenvironment{pfclma}{\begin{trivlist}\item[]{{\textbf{(Proof of Claim)}}}}{\hfill {\mbox{$\dashv$}}\end{trivlist}}

\title{Uniform Interpolation of Basic Tense Logic}
\author{Katsuhiko Sano
\institute{Faculty of Humanities and Human Sciences,
Hokkaido University
\\ Sapporo, Japan}
\email{v-sano@let.hokudai.ac.jp}
}

\newcommand{\titlerunning}{Uniform Interpolation of Basic Tense Logic}
\newcommand{\authorrunning}{K.~Sano}

\hypersetup{
  bookmarksnumbered,
  pdftitle    = {\titlerunning},
  pdfauthor   = {\authorrunning},
  pdfsubject  = {EPTCS},               
   pdfkeywords = {Uniform Interpolation, Converse Modality, Bisimulation, Basic Tense Logic, Bounded Bisimulation, Craig Interpolation}  
}

\begin{document}
\maketitle

\begin{abstract}
This paper establishes the uniform interpolation theorem for basic tense logic, which is also known as two-way modal logic or modal logic with converse. First introduced by Arthur Prior, basic tense logic is a syntactic expansion of basic modal logic with a converse modality. Its corresponding accessibility relation is defined as the converse of the standard accessibility relation in a given Kripke model. Although basic tense logic has been widely studied since its introduction, its uniform interpolation property has yet to be fully established. For basic modal logic $\mathbf{K}$, Albert Visser (1996) provided a semantic argument formulated in terms of layered (or bounded) bisimulation, explicitly attributing the uniform interpolation property of $\mathbf{K}$ to Silvio Ghilardi. This paper extends Visser's semantic argument to demonstrate that basic tense logic also enjoys the uniform interpolation property.
\end{abstract}

\section{Introduction}

Since its introduction by Arthur Prior in the 1950s~\cite{Prior1957,Prior1967,Prior1968,Prior2003}, basic tense logic has been a subject of intense study in logic, philosophy, formal semantics, and computer science, as evidenced by a vast body of literature including~\cite{Kamp1968,Mcarthur1976,Burgess1984,Blackburn1990,Benthem1991,Goldblatt1992,GHR1994}. 
Tense logic arises from a temporal interpretation of modal logic, where the necessity operator $\Box$ is interpreted as a future necessity operator, read as ``it will always be the case that'' (originally denoted by the symbol ``${G}$''). 
To obtain tense logic syntactically, a converse modality $\bbox$ is added to basic modal logic. 
This modality is interpreted as a past necessity operator read as ``it has always been the case that'' (originally denoted by the symbol ``${H}$''). Key interaction axioms between $\Box$ and $\bbox$ are $p \to \Box \bdia p$ and $p \to \bbox \dia p$, which jointly mean that the accessibility relation of $\bbox$ is the converse relation of the accessibility relation of $\Box$. In this sense, basic tense logic is also sometimes referred to as \emph{modal logic with converse} or \emph{two-way modal logic}. 

The \emph{Craig interpolation property} for a modal logic (possibly a tense logic or an intuitionistic logic) states that a provable implication $\varphi \to \psi$ of the logic has a formula $\theta$, called an interpolant, such that $\varphi \to \theta$ and $\theta \to \psi$ are provable in the logic and the propositional variables in $\theta$ are contained in the common vocabulary of $\varphi$ and $\psi$ (for a survey of interpolation in non-classical logic, the reader is referred to~\cite{DAgostino2008}). This property is strengthened to the \emph{uniform interpolation property} of a modal logic, which states that for a given formula $\varphi$ and a propositional variable $p$, there exists a uniform interpolant $\Exi{p}\varphi$ such that its propositional variables are a subset of the propositional variables of $\varphi$ excluding $p$ and the following equivalence holds for all $p$-free formulas $\psi$: $\varphi \to \psi$ is provable if and only if $\Exi{p}\varphi \to \psi$ is provable. It is well known that the uniform interpolation property implies the Craig interpolation property.

The Craig interpolation property has been studied proof-theoretically and semantically for basic tense logic. For proof-theoretic studies, the reader is referred to Maruyama et al.~\cite{Maruyama2001}, ten Cate et al.~\cite{tenCate2013}, Lyon et al.~\cite{Lyon2020}, and Sano and Yamasaki~\cite{Sano2020a}, all of which employ Maehara's method~\cite{Maehara1960} to compute an interpolant in a tableau or sequent calculus. For semantic studies, Wolter~\cite{Wolter1997a} provided a theorem on the failure of the Craig interpolation theorem for normal extensions of basic tense logic over unbounded linear transitive frames. It is also worth noting that we can apply a semantic sufficient condition given in~\cite[Theorem 2.5.3]{tenCate2005} by ten Cate for the Craig interpolation of normal tense logics of an elementary (or first-order definable) frame class in terms of bisimulation products and generated subframes. This establishes that the Craig interpolation of the minimal basic tense logic $\mathbf{Kt}$ (i.e., the bimodal $\mathbf{K}$ with the two interaction axioms) holds, since the above two characteristic axioms relating the future and past necessity operators define frame properties that are closed under the specified frame constructions. However, as far as the author knows, the uniform interpolation property of the minimal basic tense logic is still unknown. This is rather surprising, given the importance of basic tense logic in the literature of philosophy, computer science, formal semantics, and mathematical logic.

As with the Craig interpolation property, there have been two main approaches to tackling the uniform interpolation property of the minimal modal logic $\mathbf{K}$. In the proof-theoretic approach, B\'{i}lkov\'{a}~\cite{Bilkova2007} adapted Pitts's~\cite{Pitts1992} argument for the uniform interpolation property of intuitionistic propositional logic to the minimal modal logic $\mathbf{K}$ using sequent calculus, where the termination of the calculus is crucial to the argument. This proof-theoretic approach has since been extended to various classical and non-classical (possibly non-normal) modal logics, e.g., in~\cite{Iemhoff2019,AkbarTabatabai2022,AkbarTabatabai2024,Feree2024}.

Taking a semantic approach, Ghilardi and Zawadowski~\cite{Ghilardi1995a} combined the notion of layered bisimulation, also known as bounded bisimulation or $n$-bisimulation (cf.~\cite{BdRV2001}), with algebraic and categorical considerations, to provide a semantic proof of Pitts's~\cite{Pitts1992} result regarding the uniform interpolation property of intuitionistic propositional logic. Their semantic argument is an improvement upon Shavrukov's~\cite{Shavrukov1993} semantic argument for the uniform interpolation property of the G{\"o}del-L{\"o}b logic $\mathbf{GL}$. As explicitly acknowledged by Visser~\cite{Visser1996}, Ghilardi and Zawadowski's~\cite{Ghilardi1995a} argument can be adapted to prove the uniform interpolation property of $\mathbf{K}$. In this sense, the uniform interpolation property of $\mathbf{K}$ is attributed to Ghilardi and Zawadowski~\cite{Ghilardi1995a} (cf.~\cite[footnote 2]{Visser1996}). Visser~\cite{Visser1996} also provided the details of the model-theoretic argument in terms of layered bisimulations.
It is worth noting that Ghilardi and Zawadowski~\cite{Ghilardi1995} also established that the modal logic $\mathbf{S4}$ lacks the uniform interpolation property in terms of layered bisimulation (Visser~\cite{Visser1996preprint} also provided a version of the argument). More recently, Visser's model-theoretic argument given in~\cite{Visser1996} was extended by Kurahashi~\cite{Kurahashi2020} to prove the uniform Lyndon interpolation property, which generalizes both the uniform interpolation property and the Lyndon interpolation property by additionally accounting for the polarity of propositional variables in an interpolant. Kurahashi~\cite{Kurahashi2020} proved that the modal logics $\mathbf{K}$, $\mathbf{B}$, $\mathbf{GL}$, and $\mathbf{Grz}$ possess the uniform Lyndon interpolation property.

While the proof-theoretic approach provides an algorithm to compute a uniform interpolant, one of the merits of the semantic approach is that it reveals that the uniform interpolant in the basic modal logic $\mathbf{K}$ has a semantic clause as a bisimulation quantifier (cf.~\cite{French2006,DAgostino2006}), as previously discussed by Visser~\cite{Visser1996}. This paper proves the uniform interpolation property of the basic tense logic $\mathbf{Kt}$ by employing Visser's~\cite{Visser1996} semantic argument in terms of \emph{temporal} layered bisimulations. Consequently, we also prove that a uniform interpolant has a semantic clause as a temporal version of a bisimulation quantifier. Furthermore, we demonstrate that Ghilardi and Zawadowski's~\cite{Ghilardi1995} semantic argument (and its alternative reformulation given by Visser~\cite{Visser1996preprint}) for the failure of the uniform interpolation property in $\mathbf{S4}$ remains valid for the tense expansion of $\mathbf{S4}$.

We proceed as follows. Section \ref{sec:prelim} reviews the syntax, Kripke semantics, and Hilbert system of basic tense logic, and defines the notion of a normal tense logic. We then formulate the uniform interpolation property for normal tense logics. Section \ref{sec:bisim} introduces bisimulations and layered bisimulations. In Section \ref{sec:bep}, we establish a sufficient condition for the uniform interpolation property of a normal tense logic in terms of the bisimulation expansion property of a frame class. Section \ref{sec:uipkt} shows that the class of all frames possesses the bisimulation expansion property, and that the minimal basic tense logic has the uniform interpolation property. Section \ref{sec:uipextension} discusses the transfer of the uniform interpolation property via the mirror image, constant extension, and boxdot translation. Section \ref{sec:uipfailureS4t} provides an example of a normal tense logic that lacks the uniform interpolation property; specifically, we examine a tense expansion of the modal logic $\mathbf{S4}$. Finally, Section \ref{sec:concl} concludes the paper with directions for future research.

\section{Preliminaries}
\label{sec:prelim}

Given a countably infinite set $\mathsf{Prop}$ of propositional variables, we define the set $\mathsf{Form}$ of all formulas of basic tense logic as follows: 
\[
\mathsf{Form} \ni \varphi ::= p \,|\, \bot \,|\ \varphi \to \varphi \,|\ \Box \varphi \,|\ \bbox \varphi,
\]
where $p \in \mathsf{Prop}$. 
$\Box \varphi$ is read as ``$\varphi$ holds for all future states'' and $\bbox \varphi$ as ``$\varphi$ holds for all past states''. 
In addition to the ordinary Boolean connectives such as $\land$, $\lor$ and $\lnot$, 
we introduce $\dia \varphi$ and $\bdia \varphi$ (they are also denoted by $F \varphi$ and $P \varphi$ in the literature) as abbreviations for $\neg \Box \neg \varphi$ and $\neg \bbox \neg \varphi$ whose readings are ``$\varphi$ holds in some future state'' and ``$\varphi$ holds in some past state'' respectively. 
Given a finite set $\Delta$ of formulas, we use $\bigwedge \Delta$ and $\bigvee \Delta$ to mean the conjunction and disjunction of all formulas in $\Delta$ respectively, where 
$\bigwedge \varnothing$ := $\top$ and $\bigvee \varnothing$ := $\bot$. 

We define $\mathsf{P}(\varphi)$ as the set of all propositional variables in $\varphi$. 
When $\mathsf{P}(\varphi)$ $=$ $\varnothing$, we say that $\varphi$ is a {\em constant formula}.
Given a set $\mathsf{P} \subseteq \mathsf{Prop}$, we say that $\varphi$ is \emph{$\mathsf{P}$-free} if $\mathsf{P}(\varphi) \cap \mathsf{P} = \varnothing$. 
When $\mathsf{P}$ is a singleton $\setof{p}$, we simply say that $\varphi$ is \emph{$p$-free}. We define the {\em modal depth} $\mathtt{d}(\varphi)$ of a formula $\varphi$ inductively as follows: 
\[
\begin{array}{rllrll}
\mathtt{d}(p)& :=& 0, & \mathtt{d}(\bot)& :=& 0,\\
\mathtt{d}(\varphi_{1} \to \varphi_{2})& :=& \max \setof{ \mathtt{d}(\varphi_{1}),\mathtt{d}(\varphi_{2}) },&
\mathtt{d}(\bigcirc \varphi)& := & \mathtt{d}(\varphi)+1,\\
\end{array}
\]
where $\bigcirc \in \setof{\Box, \bbox}$. For example, $\mathtt{d}(p \to \Box \bdia p)$ = $2$. 
Given a finite set $\mathsf{P} \subseteq \mathsf{Prop}$ of propositional variables and $n \in \omega$, we define 
$\mathsf{Form}(\mathsf{P})$ and $\mathsf{Form}(\mathsf{P};n)$ by
\[
\begin{array}{rllrll}
\mathsf{Form}(\mathsf{P}) &:=& \inset{\varphi \in \mathsf{Form}}{\mathsf{P}(\varphi) \subseteq \mathsf{P}},&
\mathsf{Form}(\mathsf{P};n) &:=& \inset{\varphi \in \mathsf{Form}}{\mathsf{P}(\varphi) \subseteq \mathsf{P} \text{ and }\mathtt{d}(\varphi) \leqslant n}.\\
\end{array}
\]

We move to Kripke semantics of basic tense logic. We say that a pair $F$ $=$ $(W,R)$ is a {\em frame} if $W$ is a non-empty set of {\em states} (or \emph{moments}) and $R$ is a binary relation on $W$, i.e., $R \subseteq W \times W$. We use $\mathbb{F}$, $\mathbb{G}$, etc. to denote a class of frames. 
Define $\mathbb{F}_{\mathtt{all}}$ as the class of all frames. A {\em model} is a pair $M$ $=$ $(F,V)$ where $F$ $=$ $(W,R)$ is a frame and $V$ is a valuation function $V:\mathsf{Prop} \to \wp(W)$ where $\wp(W)$ is the powerset of $W$. A {\em pointed model} $(M,w)$ is a pair of a model $M$ $=$ $(W,R,V)$ and a state $w \in W$ of $M$. We say that a pointed model $(M,w)$ is \emph{from} a frame class $\mathbb{F}$ if 
$(W,R) \in \mathbb{F}$ where $M$ = $(W,R,V)$. Given a model $M$ = $(W,R,V)$, a state $w \in W$ and a formula $\varphi$, we define the {\em satisfaction relation} $M,w \models \varphi$ (read: ``$\varphi$ holds at $w$ in $M$'') inductively as follows:
\[
\begin{array}{lll}
M,w \models p&\iff&w \in V(p), \\
M,w \not\models\bot,&& \\
M,w \models \varphi_{1} \to \varphi_{2}&\iff& M,w \not\models \varphi_{1} \text{ or } M,w \models \varphi_{2},\\
M,w \models \Box \varphi &\iff& \text{$wRv$ implies } M,v\models \varphi, \text{ for all $v \in W$},\\
M,w \models \bbox \varphi &\iff& \text{$vRw$ implies } M,v\models \varphi, \text{ for all $v \in W$}.\\
\end{array}
\]
For a set $\Gamma$ of formulas, we write $M,w\models \Gamma$ to mean that $M,w\models \varphi$ for all $\varphi \in \Gamma$. 
Given a frame $F$ $=$ $(W,R)$, we say that $\varphi$ is {\em valid} on $F$ (notation: $F \models \varphi$) if $(F,V),w \models \varphi$ for all $w \in W$ and all valuation functions $V$ on $F$. Given a class $\mathbb{F}$ of frames, we say that $\varphi$ is {\em valid} on $\mathbb{F}$ (notation: $\mathbb{F} \models \varphi$) if $F \models \varphi$ for all frames $F \in \mathbb{F}$. 
Given a frame class $\mathbb{F}$, we define  $\Lambda_{\mathbb{F}}$ := $\inset{\varphi \in \mathbb{F}}{ \mathbb{F} \models \varphi}$, the set of all the valid formulas on the class $\mathbb{F}$. We say that a set $\Gamma$ of formulas is \emph{satisfiable} in a frame class $\mathbb{F}$ if there exist a frame $F \in \mathbb{F}$, a valuation $V$ on $F$, and a state $w \in W$ such that $(F,V), w \models \Gamma$. 
When $\Gamma$ is a singleton $\setof{\varphi}$, we simply say that $\varphi$ is {satisfiable} in $\mathbb{F}$. 
Note that $\varphi$ is {satisfiable} in $\mathbb{F}$ iff $\mathbb{F} \not\models \neg \varphi$. 

\begin{table}[]
    \centering
    \begin{tabular}{|rlrl|}
    \hline
         $(\mathtt{Taut})$& All classical tautologies &
         $(\mathtt{MP})$&  From $\varphi$ and $\varphi \to \psi$, we may derive $\psi$.\\
         \multicolumn{1}{|
         r}{$(\mathtt{UP})$}&\multicolumn{3}{l|}{From $\varphi$, we may derive $\varphi \sigma$, where $\sigma$ is a uniform substitution.}\\
         $(\mathtt{K}_{\Box})$&  $\Box(p \to q) \to (\Box p \to \Box q)$ &
         $(\mathtt{K}_{\bbox})$&  $\bbox(p \to q) \to (\bbox p \to \bbox q)$\\
         $(\Box\bdia)$ & $p \to \Box \bdia p$ &
         $(\bbox\dia)$& $p \to \bbox \dia p$\\
         $(\mathtt{Nec}_{\Box})$&  From $\varphi$, we may derive $\Box \varphi$. &
         $(\mathtt{Nec}_{\bbox})$& From $\varphi$, we may derive $\bbox \varphi$. \\

    \hline
    \end{tabular}
    \caption{Hilbert System $\mathsf{H}(\mathbf{Kt})$}
    \label{table:ax}
\end{table}

Table \ref{table:ax} provides the axiomatization of the minimal basic tense logic. The underlying idea of the axiomatization is to extend the minimal bimodal logic with two axioms $(\Box\bdia)$ and $(\bbox\dia)$, which state that the accessibility relation for $\bbox$ is the converse relation of the accessibility relation for $\Box$. 

\begin{definition}
A set $\Lambda$ of formulas is a {\em normal tense logic} if it contains all the axioms of Table \ref{table:ax} and it is closed under all the inference rules of Table \ref{table:ax}. We denote by $\mathbf{Kt}$ the smallest normal tense logic. 
Given a normal tense logic $\Lambda$ and a set $\Delta$ of formulas, we use $\Lambda \oplus \Delta$ to mean the smallest normal tense logic that contains $\Lambda  \cup \Delta$. 
\end{definition}

Given a frame class $\mathbb{F}$, we note that $\Lambda_{\mathbb{F}}$ := $\inset{\varphi \in \mathbb{F}}{ \mathbb{F} \models \varphi}$ is a normal tense logic. The following holds e.g., by~\cite[Corollary 4.36]{BdRV2001}. 

\begin{fact}
The axiomatization $\mathsf{H}(\mathbf{Kt})$ is sound and complete for the class $\mathbb{F}_{\mathtt{all}}$ of all frames, i.e.,
$\mathbf{Kt}$ $=$ $\Lambda_{\mathbb{F}_{\mathtt{all}}}$ $=$ $\inset{\varphi \in \mathsf{Form}}{\mathbb{F}_{\mathtt{all}} \models \varphi}$. 
\end{fact}

\begin{definition}
A normal tense logic $\Lambda$ has the {\em Craig interpolation property} if every $\varphi \to \psi \in \Lambda$ has a {\em Craig interpolant} $\theta$ such that $\setof{\varphi \to \theta, \theta \to \psi} \subseteq \Lambda$ and $\mathsf{P}(\theta) \subseteq \mathsf{P}(\varphi) \cap \mathsf{P}(\psi)$. 
\end{definition}

\noindent Both semantic and proof-theoretic proofs for the following are available; see, for example, \cite[Theorem 2.5.3]{tenCate2005} and \cite{Lyon2020, Sano2020a}, respectively.

\begin{fact}
\label{fact:CIPKt}
The smallest normal tense logic $\mathbf{Kt}$ has the Craig interpolation property. 
\end{fact}

\begin{definition}
A normal tense logic $\Lambda$ has the {\em uniform interpolation property} if every pair $(\varphi,p) \in \mathsf{Form} \times \mathsf{Prop}$ has a {\em $($post-$)$uniform interpolant} $\Exi{p}\varphi$ in $\Lambda$ such that the following three conditions are satisfied:
\begin{description}[itemsep=0pt, parsep=0pt, topsep=5pt]
    \item[(\texttt{variable})] $\mathsf{P}(\Exi{p}\varphi) \subseteq \mathsf{P}(\varphi) \setminus \setof{p}$. 
    \item[(\texttt{derivability})] $\varphi \to \Exi{p}\varphi \in \Lambda$. 
    \item[(\texttt{uniformity})] if $\varphi \to \psi \in \Lambda$ then $\Exi{p}\varphi \to \psi \in \Lambda$, for all $p$-free formulas $\psi$.
\end{description}
\end{definition}

It is easy to see that the converse implication of (\texttt{uniformity}) follows from (\texttt{derivability}).
We note that when a normal tense logic $\Lambda$ is inconsistent, i.e., $\bot \in \Lambda$, then it trivially has the uniform interpolation property, since $\Lambda$ $=$ $\mathsf{Form}$. 

\begin{proposition}
\label{prop:uip2cip}
If a normal tense logic $\Lambda$ has the uniform interpolation property, then it also has the Craig interpolation property. 
\end{proposition}

\begin{proof}
Suppose that $\Lambda$ has the uniform interpolation property. 
Fix any $\varphi \to \psi \in \Lambda$. 
If $\mathsf{P}(\varphi) \setminus \mathsf{P}(\psi)$ = $\varnothing$, we obtain $\mathsf{P}(\varphi) \subseteq \mathsf{P}(\psi)$ and so $\varphi$ is a Craig interpolant (both $\setof{\varphi \to \varphi, \varphi \to \psi} \subseteq \Lambda$ and $\mathsf{P}(\varphi) \subseteq \mathsf{P}(\varphi) \cap \mathsf{P}(\psi)$ hold).
Suppose that $\mathsf{P}(\varphi) \setminus \mathsf{P}(\psi)$ is non-empty. Put $\mathsf{P}(\varphi) \setminus \mathsf{P}(\psi)$ := $\setof{p_{1},\ldots,p_{n}}$. 
Our desired Craig interpolant is $\Exi{p_{1}}\cdots \Exi{p_{n}}\varphi$.  
It is easy to see that  
$\mathsf{P}(\Exi{p_{1}}\cdots \Exi{p_{n}}\varphi) \subseteq \mathsf{P}(\varphi) \cap \mathsf{P}(\psi)$ by the condition of  (\texttt{variable}).
By (\texttt{derivability}), we obtain $\varphi \to \Exi{p_{1}}\cdots \Exi{p_{n}}\varphi \in \Lambda$. 
Finally, the condition of (\texttt{uniformity}) and $\varphi \to \psi \in \Lambda$ jointly imply that $\Exi{p_{1}}\cdots \Exi{p_{n}}\varphi \to \psi \in \Lambda$, since $\psi$ is $\setof{p_{1},\ldots,p_{n}}$-free. 
\end{proof}

\section{Bisimulation and Layered Bisimulation}
\label{sec:bisim}
This section reviews the notions of (temporal) bisimulation and layered bisimulation, along with their basic properties. 
First, we introduce the following notion of bisimulation (or {\em temporal} bisimulation, see, e.g.,~\cite[Remark 2.25]{BdRV2001}). Hereafter, for basic tense logic, we shall refer to this simply as bisimulation throughout this paper.

\begin{definition}[(Temporal) Bisimulation]
Let $M$ = $(W,R, V)$ and $M'$ = $(W',R', V')$ be models and $\mathsf{P} \subseteq \mathsf{Prop}$. 
We say that a non-empty relation $Z \subseteq W \times W'$ is a {\em $\mathsf{P}$-bisimulation} between $M$ and $M'$ 
if $Z$ satisfies the following five conditions:
\begin{description}[itemsep=0pt, parsep=0pt, topsep=5pt]
    \item[(\texttt{Atom})] if $wZw'$ then ($w \in V(p)$ iff $w' \in V'(p)$) for all $p \in \mathsf{P}$;  
    \item[(\texttt{Forth})] 
    if $wZw'$ and $wRv$ then there exists $v' \in W'$ such that $w'R'v'$ and $vZv'$.  
    \item[($\mathtt{Forth}^{-1}$)] 
    if $wZw'$ and $vRw$ then there exists $v' \in W'$ such that $v'R'w'$ and $vZv'$. 
    \item[(\texttt{Back})] if $wZw'$ and $w'R'v'$ then there exists $v \in W$ such that $wRv$ and $vZv'$.  
    \item[($\mathtt{Back}^{-1}$)] if $wZw'$ and $v'R'w'$ then there exists $v \in W$ such that $vRw$ and $vZv'$. 
\end{description}
When the underlying $\mathsf{P}$ is clear from the context, we simply say that $Z$ is a {\em bisimulation}. If there exists a $\mathsf{P}$-bisimulation between $M$ and $M'$ such that $wZw'$, we write $(M,w)\bisim^{\mathsf{P}} (M',w')$ and say that $(M,w)$ and $(M',w')$ are {\em bisimilar}. 
\end{definition}

\begin{definition}
We write $(M_{1},w_{1}) \moeq^{\mathsf{P}} (M_{2},w_{2})$ to mean that 
$M_{1},w_{1} \models \varphi$ iff $M_{2},w_{2} \models \varphi$, for every formula $\varphi \in \mathsf{Form}(\mathsf{P})$, i.e., the set of all formulas $\varphi$ such that $\mathsf{P}(\varphi) \subseteq \mathsf{P}$.  
\end{definition}

It is well-known that bisimilarity implies formula equivalence (for basic modal logic, the reader is referred to~\cite[Theorem 2.20]{BdRV2001}), while the converse does not hold in general (cf.~\cite[Example 2.23]{BdRV2001}). 

\begin{proposition}[Formula Equivalence]
\label{prop:bisim2moeq}
Let $M$ = $(W,R,V)$ and $M'$ = $(W',R', V')$ be models and 
$\mathsf{P} \subseteq \mathsf{Prop}$. 
If $(M,w) \bisim^{\mathsf{P}} (M',w')$ then $(M,w) \moeq^{\mathsf{P}} (M',w')$. 
\end{proposition}

\begin{definition}
\label{dfn:bm}
Let $M$ = $(W,R, V)$ and $M'$ = $(W',R', V')$ be models. 
We say that a function $f: W \to W'$ is a {\em $\mathsf{P}$-bounded morphism} from $M$ to $M'$ if it satisfies the following four conditions:
\begin{description}[itemsep=0pt, parsep=0pt, topsep=5pt]
    \item[(\texttt{Atom})] $w \in V(p)$ iff $f(w) \in V'(p)$ for all $p \in \mathsf{P}$.
    \item[(\texttt{Forth})] $wRv$ implies $f(w)R'f(v)$. 
    \item[(\texttt{Back})] if $f(w)R'v'$ then there exists $v \in W$ such that $f(v)$ = $v'$ and $wRv$. 
    \item[($\mathtt{Back}^{-1}$)] if $v'R'f(w)$ then there exists $v \in W$ such that $f(v)$ = $v'$ and $vRw$. 
\end{description}
A mapping $f: W \to W'$ is called a \emph{bounded morphism} from $(W,R)$ to $(W',R')$ if it satisfies the two conditions other than (\texttt{Atom}).
\end{definition}

\noindent We note that $f: W \to W'$ is a $\mathsf{P}$-bounded morphism from $M$ to $M'$ iff 
$Gr(f)$ := $\inset{(w,f(w))}{w \in W}$ is a $\mathsf{P}$-bisimulation. Therefore, for a $\mathsf{P}$-bounded morphism $f$ from $M$ to $M'$, Proposition \ref{prop:bisim2moeq} enables us to obtain the equivalence: $M,w \models \varphi$ iff $M',f(w) \models \varphi$ for all formulas $\varphi \in \mathsf{Form}(\mathsf{P})$.

The following definition is a temporal generalization of the notion of \emph{layered bisimulation} (also known as \emph{bounded bisimulation}) introduced in~\cite[Section 1]{Ghilardi1995}, \cite[Section 3]{Visser1996}, and~\cite{Shavrukov1993}. 

\begin{definition}[Layered (Temporal) Bisimulation]
Let $M$ = $(W,R, V)$ and $M'$ = $(W',R', V')$ be models and $\mathsf{P}\subseteq \mathsf{Prop}$.
We say that a non-empty relation $Z \subseteq W \times \omega \times W$ is a {\em layered $\mathsf{P}$-bisimulation} between $M$ and $M'$ 
if $Z$ satisfies the following five conditions, where we write $wZ_{n}w'$ to mean $(w,n,w') \in Z$ below: 
\begin{description}[itemsep=0pt, parsep=0pt, topsep=5pt]
    \item[(\texttt{Atom})] if $wZ_{n}w'$ then ($w \in V(p)$ iff $w' \in V'(p)$) for all $p \in \mathsf{P}$ ($n \in \omega$).
    \item[(\texttt{Forth})] 
    if $wZ_{n+1}w'$ and $wRv$ then there exists $v' \in W'$ such that $w'R'v'$ and $vZ_{n}v'$ ($n \in \omega$).  
    \item[($\mathtt{Forth}^{-1}$)] 
    if $wZ_{n+1}w'$ and $vRw$ then there exists $v' \in W'$ such that $v'R'w'$ and $vZ_{n}v'$ ($n \in \omega$).  
    \item[(\texttt{Back})] if $wZ_{n+1}w'$ and $w'R'v'$ then there exists $v \in W$ such that $wRv$ and $vZ_{n}v'$ ($n \in \omega$).  
    \item[($\mathtt{Back}^{-1}$)] 
    if $wZ_{n+1}w'$ and $v'R'w'$ then there exists $v \in W$ such that $vRw$ and $vZ_{n}v'$ ($n \in \omega$).  
\end{description}
Let $n \in \omega$. If there exists a layered $\mathsf{P}$-bisimulation between $M$ and $M'$ such that $wZ_{n}w'$, we write $(M,w) \bisim_{n}^\mathsf{P} (M',w')$. 
We say that a layered bisimulation $Z$ between $M$ and $M'$ is {\em downward closed} if, for all $n \in \omega$, 
$wZ_{n}w'$ implies that $wZ_{m}w'$ for all $m \leqslant n$. 
\end{definition}

We can establish the following, 
similarly to~\cite[Proposition 2.29]{BdRV2001} for basic modal logic.  
The following can be also established in terms of \emph{normal forms}~\cite{Fine1975,Moss2007} with the \emph{cover modality} or \emph{nabla modality}. The reader can find an outline of the argument in Appendix \ref{sec:canfor}.

\begin{proposition}
\label{prop:foreqfindep}
The quotient set 
$\mathsf{Form}(\mathsf{P};n) / \equiv$ is finite, where $\equiv$ is logical equivalence. 
\end{proposition}

\begin{definition}
Let $\mathsf{P} \subseteq \mathsf{Prop}$. 
We define $\Phi(\mathsf{P};n)$ as a finite set of all fixed representatives of the equivalence classes in $\mathsf{Form}(\mathsf{P};n) / \equiv$ by Proposition \ref{prop:foreqfindep}. 
Given a pointed model $(M,w)$, we define the \emph{characteristic formula} $\chi^{\mathsf{P};n}_{(M,w)}$ with respect to $(\mathsf{P},n)$ by $\bigwedge \inset{ \varphi\in \Phi(\mathsf{P};n)}{ M, w \models \varphi}$. 
Let $(M_{i},w_{i})$ be a pointed model $(i \in \setof{1,2})$. 
We write $(M_{1},w_{1}) \moeq^{\mathsf{P}}_{n} (M_{2},w_{2})$ to mean that 
$M_{1},w_{1} \models \varphi$ iff $M_{2},w_{2} \models \varphi$, for every formula $\varphi \in \mathsf{Form}(\mathsf{P}; n)$. 
\end{definition}

\begin{proposition}
\label{prop:nbismnmoeq}
Let $\mathsf{P} \subseteq \mathsf{Prop}$ be finite and $n \in \omega$. For every two pointed models $(M_{1},w_{1})$ and $(M_{2},w_{2})$, the following are all equivalent: 
\begin{enumerate*}[label=\arabic*., ref=\arabic*]
    \item \label{item:nmoeq} $(M_{1},w_{1}) \moeq_{n}^{\mathsf{P}} (M_{2},w_{2})$; 
    \item \label{item:nmoeqchar}$M_{1},w_{1} \models \chi_{(M_{2},w_{2})}^{\mathsf{P};n}$;  
    \item \label{item:nlaybisim} There exists a downward layered $\mathsf{P}$-bisimulation $Z$ between $M_{1}$ and $M_{2}$ such that $(w_{1},n, w_{2}) \in Z$;  
    \item \label{item:nbisim} $(M_{1},w_{1}) \bisim_{n}^{\mathsf{P}} (M_{2},w_{2})$.
\end{enumerate*}
\end{proposition}

\begin{proof}
The equivalence between items \ref{item:nmoeq} and \ref{item:nmoeqchar} is immediate by definition of $\chi_{(M_{2},w_{2})}^{\mathsf{P};n}$.  
The direction from item \ref{item:nmoeq} to item \ref{item:nlaybisim} holds by defining $(w,n,w') \in Z$ iff 
$(M_{1},w_{1}) \moeq_{n}^{\mathsf{P}} (M_{2},w_{2})$. 
The direction from item \ref{item:nlaybisim} to item \ref{item:nbisim} is trivial. 
Finally, the direction from item \ref{item:nbisim} to item \ref{item:nmoeq} is proved by induction on $\varphi \in \mathsf{Form}(\mathsf{P};n)$. 
\end{proof}

\section{Bisimulation Expansion Property}
\label{sec:bep}
This section formulates the bisimulation expansion property of a frame class $\mathbb{F}$, which provides a sufficient condition for the uniform interpolation property of the normal tense logic $\Lambda_{\mathbb{F}}$ associated with that class. Furthermore, we establish that the uniform interpolant in $\Lambda_{\mathbb{F}}$ coincides with a bisimulation quantifier, provided that the frame class $\mathbb{F}$ has the bisimulation expansion property.


While the following definition is based on \cite[Lemma 6.1]{Visser1996} and \cite[Definition 11]{Kurahashi2020}, we introduce four improvements: first, we refine the number of sets of propositional variables (whereas \cite[Lemma 6.1]{Visser1996} requires three pairwise disjoint sets); second, we shift the focus from model classes to frame classes; third, we prioritize bisimilarity over formula equivalences up to a given modal depth; and finally, we lift the finiteness restriction on $\mathsf{P}_{2}$ below (cf.~\cite{Kurahashi2020}).

\begin{definition}
\label{dfn:bep}
A frame class $\mathbb{F}$ has the {\em bisimulation expansion property} if, for all $\mathsf{P}_{1}, \mathsf{P}_{2} \subseteq \mathsf{Prop}$ such that $\mathsf{P}_{1}$ is finite and for all pointed models $(M_{1},w_{1})$ and $(M_{2},w_{2})$ from $\mathbb{F}$ and $n \in \omega$ such that $(M_{1},w_{1}) \bisim_{n}^{\mathsf{P}_{1} \cap \mathsf{P}_{2}} (M_{2},w_{2})$, there exists a pointed model $(N,x)$ from $\mathbb{F}$ such that $(M_{1},w_{1}) \bisim_{n}^{\mathsf{P}_{1} }(N,x) \bisim^{\mathsf{P}_{2}} (M_{2},w_{2})$.
\end{definition}


\begin{lemma}
\label{lem:bep2uip}
If a frame class $\mathbb{F}$ has the bisimulation expansion property, then 
$\Lambda_{\mathbb{F}}$ $:=$ $\inset{\varphi \in \mathsf{Form}}{\mathbb{F} \models \varphi}$ has the uniform interpolation property. 
\end{lemma}

\begin{proof}
Suppose that a frame class $\mathbb{F}$ has the bisimulation expansion property.
Fix any pair $(\varphi,p) \in \mathsf{Form} \times \mathsf{Prop}$. 
Put $n$ := $\mathtt{d}(\varphi)$. 
Recall that $\Phi(\mathsf{P}(\varphi)\setminus \setof{p};n)$ is a finite set of all fixed representatives of the equivalence classes in $\mathsf{Form}(\mathsf{P}(\varphi)\setminus \setof{p};n) / \equiv$.
Define $\mathsf{Cn}(\varphi)$ := $\inset{\psi \in  \Phi(\mathsf{P}(\varphi)\setminus \setof{p};n)}{\mathbb{F} \models \varphi \to \psi}$. 
Then we define $\Exi{p}\varphi$ := $\bigwedge \mathsf{Cn}(\varphi)$. 
In what follows, we prove that $\Exi{p}\varphi$ (= $\bigwedge \mathsf{Cn}(\varphi)$) satisfies all the required conditions. 
The condition (\texttt{variable}) holds by definition. 
For (\texttt{derivability}), we show that $\mathbb{F} \models \varphi \to \bigwedge \mathsf{Cn}(\varphi)$. 
Fix any $\psi \in  \Phi(\mathsf{P}(\varphi) \setminus \setof{p};n)$ such that $\mathbb{F} \models \varphi \to \psi$. It suffices to show that $\mathbb{F} \models \varphi \to \psi$, which is immediate. 

So, in what follows, we focus on establishing the condition $(\texttt{uniformity})$. Fix any $p$-free formula $\psi$ (i.e., $p \notin \mathsf{P}(\psi)$) and suppose that $\mathbb{F} \models \varphi \to \psi$. 
We show that $\mathbb{F} \models \Exi{p}\varphi\to \psi$. 
Let us fix any pointed model $(M_{2},w_{2})$ from $\mathbb{F}$ such that $M_{2},w_{2} \models \Exi{p}\varphi$. Our goal is to show that $M_{2},w_{2}\models \psi$. Define $\chi := \chi^{\mathsf{Q};n}_{(M_{2},w_{2})}$ where 
$\mathsf{Q}$ $:=$ $(\mathsf{P}(\varphi) \setminus \setof{p}) \cap \mathsf{P}(\psi)$
Then $\varphi \land \chi$ is satisfiable in $\mathbb{F}$, i.e., $\mathbb{F} \not\models \neg(\varphi \land \chi)$. This is verified as follows: Suppose not, i.e., $\mathbb{F} \models \neg(\varphi \land \chi)$. It follows that $\mathbb{F} \models \varphi \to \neg \chi$. 
Since $\neg \chi$ is $p$-free, we obtain $\neg \chi \in \mathsf{Cn}(\varphi)$ hence $\mathbb{F} \models \Exi{p}\varphi \to \neg \chi$ by $\Exi{p}\varphi$ = $\bigwedge \mathsf{Cn}(\varphi)$.  
Since $M_{2},w_{2} \models \Exi{p}\varphi$, we get $M_{2},w_{2} \models \neg \chi$, which is a contradiction with $M_{2},w_{2} \models \chi$. 
By the satisfiability of $\varphi \land \chi$ in $\mathbb{F}$, we can find a pointed model $(M_{1},w_{1})$ from $\mathbb{F}$ such that $M_{1},w_{1} \models \varphi \land \chi$. 
We note that $\mathsf{Q}$ $=$ $(\mathsf{P}(\varphi) \setminus \setof{p}) \cap \mathsf{P}(\psi)$ $=$ $\mathsf{P}(\varphi) \cap \mathsf{P}(\psi)$, since $\psi$ is $p$-free. 
It follows from $M_{1},w_{1} \models \chi$ that 
$(M_{1},w_{1}) \bisim_{n}^{\mathsf{P}(\varphi) \cap \mathsf{P}(\psi)}  (M_{2},w_{2})$. By the bisimulation expansion property of $\mathbb{F}$, we can find a pointed model $(N,x)$ from $\mathbb{F}$ such that $(M_{1},w_{1}) \bisim_{n}^{\mathsf{P}(\varphi) }(N,x) \bisim^{\mathsf{P}(\psi)} (M_{2},w_{2})$. 
It follows from $M_{1},w_{1} \models \varphi$ and $(M_{1},w_{1}) \bisim_{n}^{\mathsf{P}(\varphi) }(N,x)$ that $N,x \models \varphi$ (recall that $n$ = $\mathtt{d}(\varphi)$). By the initial supposition $\mathbb{F} \models \varphi \to \psi$, we obtain $N,x \models \psi$.
Therefore, we derive from $(N,x) \bisim^{\mathsf{P}(\psi)} (M_{2},w_{2})$ that $M_{2},w_{2} \models \psi$, as desired. 
\end{proof}

\begin{remark}
Even if we restrict our focus to formula equivalences up to a given modal depth and retain the finiteness restriction on $\mathsf{P}_{2}$ in Definition \ref{dfn:bep}, we can still provide a sufficient condition for the uniform interpolation property of $\Lambda_{\mathbb{F}}$. In the proof of Corollary \ref{cor:bisimquan}, however, it is crucial to allow the set $\mathsf{P}_{2}$ of propositional variables to be potentially infinite within the bisimulation expansion property. In this sense, Lemma \ref{lem:bep2uip} has broader applicability.
\end{remark}

The following provides a generalization of Visser's argument for the definability of bisimulation quantifiers in the basic modal logic $\mathbf{K}$ (cf.~\cite[Theorem 6.2]{Visser1996}) to a normal tense logic $\Lambda_{\mathbb{F}}$. 

\begin{corollary}
\label{cor:bisimquan}
If a frame class $\mathbb{F}$ has the bisimulation expansion property, then the following equivalence holds for a uniform interpolant $\Exi{p}\varphi$ for $(\varphi,p) \in \mathsf{Form} \times \mathsf{Prop}$: 
for every $(M,w)$ from $\mathbb{F}$, 
\[
\begin{array}{lll}
     M,w \models \Exi{p}\varphi &\iff& 
     (M,w) \bisim^{\mathsf{Prop}\setminus \setof{p} } (N,v) \text{ and } N,v \models \varphi, 
     \text{ for some $(N,v)$ from $\mathbb{F}$. }
     \\
\end{array}
\]
\end{corollary}

\begin{proof}
Suppose that a frame class $\mathbb{F}$ has the bisimulation expansion property. By Lemma \ref{lem:bep2uip}, $\Lambda_{\mathbb{F}}$ has the uniform interpolation property. 
Fix any pointed model $(M,w)$. 
For the right-to-left direction, let us suppose that we can find a pointed model $(N,v)$ from $\mathbb{F}$ such that $(M,w) \bisim^{\mathsf{Prop}\setminus \setof{p} } (N,v)$ and $N,v \models \varphi$. Our goal is to show that $M,w \models \Exi{p}\varphi$. Since $\Exi{p}\varphi$ is a uniform interpolant for $(\varphi,p)$, we obtain $\mathbb{F} \models \varphi \to \Exi{p}\varphi$ by $(\mathtt{derivability})$. It follows from $N,v \models \varphi$ that $N,v \models \Exi{p}\varphi$. By $(\mathtt{variable})$ condition, $ \Exi{p}\varphi$ is $p$-free. We derive from $(M,w) \bisim^{\mathsf{Prop}\setminus \setof{p} } (N,v)$ that $M,w \models \Exi{p}\varphi$, as desired. 

For the left-to-right direction, assume that $M,w \models \Exi{p}\varphi$. Put $n$ := $\mathtt{d}(\varphi)$. 
We define 
\begin{center}
$\chi$ := $\chi^{\mathsf{P}(\varphi)\setminus \setof{p};n}_{(M,w)}$ = $\bigwedge \inset{\rho \in \Phi(\mathsf{P}(\varphi)\setminus \setof{p};n)}{M,w\models \rho}$. 
\end{center}
Note that $\chi$ is $p$-free and $M,w \models \chi$ hence $(\Exi{p}\varphi) \land \chi$ is satisfiable in $\mathbb{F}$. We prove that $\varphi\land  \chi$ is satisfiable in $\mathbb{F}$, i.e., $\mathbb{F} \not\models \neg(\varphi  \land  \chi)$. This is proved as follows: Suppose otherwise. Then $\mathbb{F} \models \varphi \to \neg \chi$. 
Since $\Exi{p}\varphi$ is a uniform interpolant for $(\varphi,p)$ and $\neg \chi$ is $p$-free, we obtain $\mathbb{F} \models \Exi{p} \varphi \to \neg \chi$ by $(\mathtt{uniformity})$. This means that $(\Exi{p}\varphi) \land \chi$ is not satisfiable in $\mathbb{F}$. A contradiction. Thus, $\varphi\land  \chi$ is satisfiable in $\mathbb{F}$. 

We can find a pointed model $(M_{0},w_{0})$ from $\mathbb{F}$ such that $M_{0}, w_{0} \models \varphi$ and $M_{0}, w_{0} \models \chi$. By the latter, we obtain $(M_{0},w_{0}) \bisim_{n}^{\mathsf{P}(\varphi)\setminus \setof{p}} (M,w)$. 
Since $\mathsf{P}(\varphi)\setminus \setof{p}$ = $\mathsf{P}(\varphi) \cap (\mathsf{Prop}\setminus \setof{p})$, the bisimulation expansion property of $\mathbb{F}$ enables us to find a pointed model $(N,x)$ from $\mathbb{F}$ such that 
\[
(M_{0},w_{0}) \bisim_{n}^{\mathsf{P}(\varphi)} (N,x) \bisim^{\mathsf{Prop} \setminus \setof{p}} (M,w).
\]
By $M_{0}, w_{0} \models \varphi$, we obtain $(N,x) \models \varphi$. By $(N,x) \bisim^{\mathsf{Prop} \setminus \setof{p}} (M,w)$, we obtain our desired goal. 
\end{proof}

\section{Uniform Interpolation Property of Basic Tense Logic}
\label{sec:uipkt}

This section establishes that the class $\mathbb{F}_{\mathtt{all}}$ of all frames has the bisimulation expansion property hence $\mathbf{Kt}$ ($=$ $\inset{\varphi \in\mathsf{Form}}{\mathbb{F}_{\mathtt{all}} \models \varphi}$) has the uniform interpolation property. 

\begin{lemma}
\label{lem:bep}
The class $\mathbb{F}_{\mathtt{all}}$ of all frames has the bisimulation expansion property. 
\end{lemma}


\begin{proof}
Fix any $\mathsf{P}_{1}, \mathsf{P}_{2} \subseteq \mathsf{Prop}$ such that $\mathsf{P}_{1}$ is finite and any $n \in \omega$. 
Let us also fix any pointed models $(M_{1},w_{1})$ and $(M_{2},w_{2})$ from $\mathbb{F}_{\mathtt{all}}$ with $(M_{1},w_{1}) \bisim_{n}^{\mathsf{P}_{1} \cap \mathsf{P}_{2}} (M_{2},w_{2})$. 
Let us put $M_{i}$ := $(W_{i},R_{i},V_{i})$. 
Since $\mathsf{P}_{1}\cap \mathsf{P}_{2}$ is finite, Proposition \ref{prop:nbismnmoeq} enables us to find a downward-closed layered $\mathsf{P}_{1}\cap \mathsf{P}_{2}$-bisimulation $Z \subseteq W_{1} \times \omega \times W_{2}$ such that $(w_{1},n,w_{2}) \in Z$. We add a new state ``$\mathtt{T}$'' to $M_{1}$ to define $M_{1}^{\mathtt{T}}$ := $(W_{1}^{\mathtt{T}},R_{1}^{\mathtt{T}}, V_{1}^{\mathtt{T}})$ as follows:
\begin{center}
$W_{1}^{\mathtt{T}}$ := $W_{1} \cup \setof{\mathtt{T}}$; \quad $R_{1}^{\mathtt{T}}$ := $R_{1} \cup \inset{(\mathtt{T},x), (x,\mathtt{T})}{x \in W_{1} \cup \setof{\mathtt{T}}}$; \quad $V_{1}^{\mathtt{T}}(p)$ := $V_{1}(p)$ for all $p \in \mathsf{Prop}$.
\end{center}
By noting that $\{-1\} \cup \omega$  (= $1+\omega$) and $\omega$ are order-isomorphic as well-ordered sets, we revise $Z$ into $Z^{+} $ as follows: $Z^{+} := Z \cup \inset{(\mathtt{T},-1,y)}{y \in W_{2}}$. We also define the \emph{predecessor function} $\mathtt{pd}: \{-1\} \cup \omega \to \{-1\} \cup \omega$ by 
$\mathtt{pd}(\alpha)$ $:=$  $\alpha -1 $ (if $\alpha \geqslant 1$); $-1$ (if $\alpha  \in \setof{-1,0}$).  
Now we define our intended intermediate model $N$ $=$ $(Z^{+},S,U)$ as follows:
\begin{itemize}
    \item $(a_{1},\alpha ,a_{2})S(b_{1},\beta,b_{2})$ iff 
    $a_{1}R_{1}^{\mathtt{T}} b_{1}$, $a_{2}R_{2}b_{2}$ and ($\mathtt{pd}(\alpha)$ = $\beta$ or $\mathtt{pd}(\beta)$ = $\alpha$). 
    \item $U: \mathsf{P}_{1} \cup \mathsf{P}_{2} \to \wp(Z^{+})$ is defined by 
    \[
    \begin{array}{lll}
    (a_{1},\alpha,a_{2}) \in U(p)&\iff& 
    \begin{cases}
     a_{1} \in V_{1}(p)   & \text{ if $p \in \mathsf{P}_{1}\setminus \mathsf{P}_{2}$;} \\
      a_{2} \in V_{2}(p)  & \text{ if $p \in \mathsf{P}_{2}$.}
    \end{cases}
    \end{array}
    \]
\end{itemize}
In what follows, we establish $(M_{1},w_{1}) \bisim^{\mathsf{P}_{1}}_{n} (N, (w_{1},n,w_{2}))
\bisim^{\mathsf{P}_{2}} (M_{2},w_{2})$. Define the projections $\pi_{i}$ ($i \in \setof{1,2,3}$) from $Z^{+}$ to each component by the following: for any $(a_{1},\alpha,a_{2}) \in Z^{+}$, 
\[
\begin{array}{rrlrrlrrl}
     \pi_{1}: (a_{1},\alpha,a_{2})&\mapsto& a_{1};  &
     \pi_{2}: (a_{1},\alpha,a_{2})&\mapsto& \alpha; &
     \pi_{3}: (a_{1},\alpha,a_{2})&\mapsto& a_{2}.
\end{array}
\]
Given an element $\alpha \in \omega$, 
we define $Z(\alpha)$ := $\pi_{2}^{-1}[\setof{\alpha}]$, the inverse image of $\setof{\alpha}$ under the projection $\pi_{2}$. 
Note that we do not need to consider $-1$ as an index in the next claim. 
Define $Y_{\alpha}$ := $Gr(\pi_{1} \upharpoonright Z(\alpha)$ ($\alpha \in \omega$), where $Gr(f)$ = $\inset{(x,f(x))}{x \in X}$ is a graph of a mapping $f:X \to Y$ and $f\upharpoonright X'$ is the restriction of $f$ to $X' \subseteq X$.

\begin{clma1}
\label{claim:m1}
$(Y_{\alpha})_{\alpha \in \omega}$ is a layered $\mathsf{P}_{1}$-bisimulation between $N$ and $M_{1}$.  
\end{clma1}

\begin{pfclma}
We establish $(\mathtt{Atom})$, $(\mathtt{Forth}^{-1})$ and 
$(\mathtt{Back}^{-1})$ since 
$(\mathtt{Forth})$ and 
$(\mathtt{Back})$ are shown similarly to 
$(\mathtt{Forth}^{-1})$ and 
$(\mathtt{Back}^{-1})$, respectively. 
\begin{itemize}
    \item \fbox{$\mathtt{Atom}$} Fix any $\alpha \in \omega$ and $p \in \mathsf{P}_{1}$. Suppose that $\pi_{1}(a_{1},\alpha,a_{2})$ = $a_{1}$. 
We show that $\pi_{1}(a_{1},\alpha,a_{2}) \in V_{1}(p)$ iff $(a_{1},\alpha,a_{2}) \in U(p)$. When $p \in \mathsf{P}_{1} \setminus \mathsf{P}_{2}$, the equivalence is immediate by definition. So, we let $p \in \mathsf{P}_{1}\cap \mathsf{P}_{2}$. 
Then $(a_{1},\alpha,a_{2}) \in U(p)$ iff $a_{2} \in V_{2} (p)$ by definition. Since $(a_{1},\alpha,a_{2}) \in Z$, $a_{2} \in V_{2} (p)$ is equivalent with $a_{1} \in V_{1}(p)$, as desired. 
\item \fbox{$\mathtt{Forth}^{-1}$} Assume  $(a_{1}, \alpha+1,a_{2}) Y_{\alpha+1} a_{1}$ and $Z^{+}\ni(b_{1},\beta,b_{2})S(a_{1}, \alpha+1,a_{2})$. 
It follows from the definition of $S$ that $b_{1}R_{1}^{\mathtt{T}}a_{1}$ hence 
$b_{1}R_{1}a_{1}$ by $\alpha + 1, \beta \in \omega$. 
In what follows, it suffices to establish $(b_{1},\alpha,b_{2})Y_{\alpha} b_{1}$. 
By assumption, we obtain i) $\mathtt{pd}(\alpha + 1)$ = $\beta$ or ii) $\mathtt{pd}(\beta)$ = $\alpha +1$. 
When i)  $\mathtt{pd}(\alpha + 1)$ = $\beta$ holds, we get $\beta$ = $\alpha$. Since $(b_{1},\alpha,b_{2}) \in Z(\alpha)$, we obtain  $(b_{1},\alpha,b_{2}) Y_{\alpha} b_{1}$. 
When ii) $\mathtt{pd}(\beta)$ = $\alpha +1$ holds, we get $\beta$ = $\alpha + 2$. Thus, $(b_{1},\alpha+2,b_{2}) \in Z$ by $\alpha \in \omega$. By the downward closure of $Z$, we obtain $(b_{1},\alpha,b_{2}) \in Z$. Therefore, we get $(b_{1},\alpha,b_{2}) Y_{\alpha} b_{1}$. 

\item \fbox{$\mathtt{Back}^{-1}$} Assume that $(a_{1}, \alpha+1,a_{2})  Y_{\alpha + 1} a_{1}$ and $b_{1} R_{1}a_{1}$. 
By $(a_{1}, \alpha+1, a_{2}) \in Z$ and $b_{1}R_{1}a_{1}$, we can find $b_{2} \in W_{2}$ such that $b_{2}R_{2}a_{2}$ and $(b_{1},\alpha,b_{2}) \in Z$ by $(\mathtt{Forth}^{-1})$ of $Z$. Then $(b_{1}, \alpha,b_{2}) Y_{\alpha} b_{1}$ and $(b_{1}, \alpha,b_{2}) S (a_{1}, \alpha + 1,a_{2})$ by $\mathtt{pd}(\alpha +1)$ = $\alpha$. 
\end{itemize}
This finishes establishing $(Y_{\alpha})_{\alpha \in \omega}$ is a layered $\mathsf{P}_{1}$-bisimulation between $N$ and $M_{1}$. 
\end{pfclma}

\begin{clma}
\label{claim:m2}
$\pi_{3}$ is a {$\mathsf{P}_{2}$}-bounded morphism from $N$ to $M_{2}$. 
\end{clma}

\begin{pfclma}
By Definition \ref{dfn:bm}, we check the following four conditions. 
\begin{itemize}
    \item \fbox{$\mathtt{Atom}$} Fix any $p \in \mathsf{P}_{2}$ and $(a_{1},\alpha,a_{2}) \in Z^{+}$. 
    We show that $\pi_{
    3}(a_{1},\alpha,a_{2}) \in V_{2}(p)$ iff $(a_{1},\alpha,a_{2}) \in U(p)$. This is immediate by definition of $U$. 
    \item \fbox{$\mathtt{Forth}$}
    Assume $(a_{1},\alpha,a_{2})S(b_{1},\beta,b_{2})$. It follows that $a_{2}R_{2}b_{2}$ hence $\pi_{3}(a_{1},\alpha,a_{2})R_{2}\pi_{3}(b_{1},\beta,b_{2})$. 
    \item \fbox{$\mathtt{Back}$} 
    Fix any $(a_{1},\alpha,a_{2}) \in Z^{+}$. Assume that $a_{2}$ $=$  $\pi_{3}(a_{1},\alpha,a_{2})R_{2}b_{2}$. We divide our argument into the following two cases: i) $\alpha \geqslant 1$; ii) $\alpha \in \setof{-1,0}$. 
    For the case i), we proceed as follows: By $\alpha \geqslant 1$, we obtain $(a_{1},\alpha,a_{2}) \in Z$. 
    By $a_{2}R_{2}b_{2}$ and $(\mathtt{Back})$ of $Z$, we can find some $b_{1} \in W_{1}$ such that $a_{1}R_{1}b_{1}$ and $(b_{1},\alpha-1,b_{2}) \in Z$. So, 
    $(a_{1},\alpha,a_{2})S(b_{1},\alpha-1,b_{2})$ and $\pi_{3}(b_{1},\alpha-1,b_{2})$ $=$ $b_{2}$. 
    Finally, we consider the case ii): $\alpha \in \setof{-1,0}$. 
    By definition of $S$, $(a_{1},\alpha,a_{2})S(\mathtt{T},-1,b_{2})$ and $\pi_{3}(\mathtt{T},-1,b_{2})$ $=$ $b_{2}$, as required. 
    \item \fbox{$\mathtt{Back}^{-1}$}
    Fix any $(a_{1},\alpha,a_{2}) \in Z^{+}$. Assume that   $b_{2}R_{2}\pi_{3}(a_{1},\alpha,a_{2})$ = $a_{2}$. 
    We divide our argument into the following two cases: i) $\alpha \geqslant 1$; ii) $\alpha \in \setof{-1,0}$. 
    For the case i), we proceed as follows: By $\alpha \geqslant 1$, we obtain $(a_{1},\alpha,a_{2}) \in Z$. 
    By $b_{2}R_{2}a_{2}$ and $(\mathtt{Back}^{-1})$ of $Z$, we can find some $b_{1} \in W_{1}$ such that $b_{1}R_{1}a_{1}$ and $(b_{1},\alpha-1,b_{2}) \in Z$. So, 
    $(b_{1},\alpha-1,b_{2})S(a_{1},\alpha,a_{2})$ and $\pi_{3}(b_{1},\alpha-1,b_{2})$ $=$ $b_{2}$. 
    Finally, we consider the case ii): $\alpha \in \setof{-1,0}$. 
    By definition of $S$, $(\mathtt{T},-1,b_{2})S(a_{1},\alpha,a_{2})$ and $\pi_{3}(\mathtt{T},-1,b_{2})$ $=$ $b_{2}$, as required.
\end{itemize}
This finishes establishing that $\pi_{3}$ is a $\mathsf{P}_{2}$-bounded morphism from $N$ to $M_{2}$. 
\end{pfclma}

It is easy to see that $((w_{1},n,w_{2}),w_{1}) \in Gr(\pi_{1} \upharpoonright Z(n) )$ and that $\pi_{3}: (w_{1},n,w_{2}) \mapsto w_{2}$. 
By Claims \ref{claim:m1} and \ref{claim:m2}, we have established $(M_{1},w_{1}) \bisim^{\mathsf{P}_{1}}_{n} (N, (w_{1},n,w_{2}))
\bisim^{\mathsf{P}_{2}} (M_{2},w_{2})$.
\end{proof}

By Lemmas \ref{lem:bep2uip} and \ref{lem:bep}, we obtain our main theorem of this paper as follows:

\begin{theorem}
\label{thm:uipkt}
The smallest normal tense logic $\mathbf{Kt}$ has the uniform interpolation property. 
\end{theorem}

\section{Uniform Interpolation Property of Normal Tense Extensions}
\label{sec:uipextension}

This section begins with several examples of frame definability to introduce various extensions of $\mathbf{Kt}$. 
We first discuss the notion of mirror images and prove that for a normal tense logic $\Lambda$ invariant under the mirror image, the uniform interpolation property of $\Lambda \oplus \Delta$ implies that of $\Lambda$ extended by the mirror image of $\Delta$. 
Subsequently, we demonstrate that the uniform interpolation property is preserved under the addition of constant axioms. 
Finally, we establish that this preservation also holds for the reflexive counterpart of a normal tense logic under certain assumptions.

Let us say that a formula set $\Gamma$ defines a frame class $\mathbb{F}$ if for all frames $F$, the following equivalence holds: $F \in \mathbb{F}$ iff $F \models \varphi$ for all $\varphi \in \Gamma$. Table \ref{table:framedef} presents examples of frame class definability.
The formulas $\mathtt{LM}$ and $\mathtt{FM}$ of Table \ref{table:framedef} express the existence of a last and a first moment (cf.~\cite{Burgess1984}), respectively. On the other hand, the formulas $\mathtt{D}_{\Box}$ and $\mathtt{D}_{\bbox}$ of Table \ref{table:framedef} express the absence of a last and a first moment, respectively. We can say that $\mathtt{LM}$ and $\mathtt{D}_{\Box}$ are \emph{mirror images} of $\mathtt{FM}$ and $\mathtt{D}_{\bbox}$, respectively, in the following sense. 

\begin{definition}
Given a formula $\varphi$, the mirror image $\mathtt{m}(\varphi)$ of $\varphi$ is defined as the formula obtained by simultaneously replacing all occurrences of $\bbox$ with $\Box$ and vice versa.
For a formula set $\Delta$, define $\mathtt{m}[\Delta]$ := $\inset{\mathtt{m}(\varphi)}{\varphi \in \Delta}$. 
\end{definition}

Given a normal modal logic $\Lambda$ and a set $\Delta$ of formulas, recall that $\Lambda \oplus \Delta$ is the smallest normal modal logics containing $\Lambda \cup \Delta$. 
The following proposition is easy to obtain. 

\begin{proposition}
\label{prop:mirror}
\begin{enumerate*}[label=\arabic*., ref=\arabic*]
    \item 
    \label{item:mirrorref} Let $\Lambda$ be a normal tense logic. Then $\Box p \to p \in \Lambda$ iff $\bbox p \to p \in \Lambda$. 
    Moreover, $\Box p \to \Box \Box p \in \Lambda$ iff $\bbox p \to \bbox \bbox p \in \Lambda$.
    \item \label{item:mirrortra} For any $\Delta \subseteq \setof{\mathtt{T}_{\Box}, \mathtt{4}_{\Box}}$, $\mathtt{m}[\mathbf{Kt}\oplus \Delta]$ $=$ $\mathbf{Kt}\oplus \Delta$. 
    \item\label{item:mirror} Let $\Lambda$ be a normal tense logic such that $\mathtt{m}[\Lambda]$ = $\Lambda$ and $\Delta$ a formula set. Then $\varphi \in \Lambda \oplus \Delta$ iff $\mathtt{m}(\varphi) \in \Lambda \oplus \mathtt{m}[\Delta]$. 
\end{enumerate*}
\end{proposition}

\begin{table}[htbp]
    \centering
\begin{tabular}{|cllcll|}
    \hline
    $(\mathtt{D}_{\Box})$ & $\dia \top$ & $\Any{x}\Exi{y}xRy$ &
    $(\mathtt{D}_{\bbox})$ & $\bdia \top$ & $\Any{x}\Exi{y}yRx$ \\
    $(\mathtt{T}_{\Box})$ & $\Box p \to p$ & $\Any{x}xRx$ &
    $(\mathtt{4}_{\Box})$ & $\Box p \to \Box\Box p$ & $\Any{x,y,z}(xRy \land yRz \to xRz)$ \\
    \hline
    \multicolumn{1}{|c}{$(\mathtt{LM})$} 
    & \multicolumn{1}{c}{$\Box \bot \lor \dia \Box \bot$}  & \multicolumn{4}{l|}{$\Any{x}(\Exi{y} (xRy) \to \Exi{z}(xRz \land \Any{w} (\neg zRw)))$} \\
    \multicolumn{1}{|c}{$(\mathtt{FM})$} & \multicolumn{1}{c}{$\bbox \bot \lor \bdia \bbox \bot$}  &  \multicolumn{4}{l|}{$\Any{x}(\Exi{y} (yRx) \to \Exi{z}(zRx \land \Any{w} (\neg wRz)))$} \\
    \hline
\end{tabular}
     \caption{Examples of Frame Definability}
    \label{table:framedef}
\end{table}

We note that the mirror image $\mathtt{m}$ does not change the set of propositional variables in a formula. 

\begin{proposition}
\label{prop:uipmirror}
Let $\Delta$ be a set of formulas and $\Lambda$ be a normal tense logic such that $\mathtt{m}[\Lambda]$ = $\Lambda$. 
If $\Lambda \oplus \Delta$ has the uniform interpolation property, then $\Lambda \oplus \mathtt{m}[\Delta]$ also has the uniform interpolation property. 
\end{proposition}

\begin{proof}
Suppose that $\Lambda \oplus \Delta$ has the uniform interpolation property. 
For the uniform interpolation property of 
$\Lambda \oplus \mathtt{m}[\Delta]$, let us fix any $(\varphi,p) \in \mathsf{Form} \times \mathsf{Prop}$. By the supposition, let $\Exi{p}\mathtt{m}(\varphi)$ be a uniform interpolant for $(\mathtt{m}(\varphi),p)$ in $\Lambda \oplus \Delta$. 
We show that $\mathtt{m}(\Exi{p}\mathtt{m}(\varphi))$ is our desired uniform interpolant for $(\varphi,p)$ in $\Lambda \oplus \mathtt{m}[\Delta]$. 
First, we show $(\mathtt{variable})$: $\mathsf{P}(\mathtt{m}(\Exi{p}\mathtt{m}(\varphi)))$ = 
$\mathsf{P}(\Exi{p}\mathtt{m}(\varphi))$ $\subseteq$ $\mathsf{P}(\mathtt{m}(\varphi)) \setminus \setof{p}$ 
$=$ $\mathsf{P}(\varphi) \setminus \setof{p}$, as desired. Next, we move to  $(\mathtt{derivability})$.
We prove that $\varphi \to \mathtt{m}(\Exi{p}\mathtt{m}(\varphi)) \in \Lambda \oplus \mathtt{m}[\Delta]$. 
Since $\Exi{p}\mathtt{m}(\varphi)$ is a uniform interpolant for $(\mathtt{m}(\varphi),p)$ in $\Lambda \oplus \Delta$, we obtain $\mathtt{m}(\varphi)  \to \Exi{p}\mathtt{m}(\varphi) \in \Lambda \oplus \Delta$ hence $\varphi \to \mathtt{m}(\Exi{p}\mathtt{m}(\varphi)) \in \Lambda \oplus \mathtt{m}[\Delta]$ by item \ref{item:mirror} of Proposition \ref{prop:mirror}. 
Finally, we establish $(\mathtt{uniformity})$. 
Fix any $p$-free formula $\psi$ and suppose that $\varphi \to \psi \in \Lambda \oplus \mathtt{m}[\Delta]$. 
Then $\mathtt{m}(\varphi) \to \mathtt{m}(\psi) \in \Lambda \oplus \Delta$ by item \ref{item:mirror} of Proposition \ref{prop:mirror}. 
Since $\Exi{p}\mathtt{m}(\varphi)$ is a uniform interpolant for $(\mathtt{m}(\varphi),p)$ in $\Lambda \oplus \Delta$ and $\mathtt{m}(\psi)$ is still $p$-free, we obtain $\Exi{p}\mathtt{m}(\varphi) \to \mathtt{m}(\psi) \in \Lambda \oplus \Delta$. We conclude from item \ref{item:mirror} of Proposition \ref{prop:mirror}  that 
$\mathtt{m}(\Exi{p}\mathtt{m}(\varphi)) \to \psi \in \Lambda \oplus \mathtt{m}[\Delta]$, as required. 
\end{proof}

\noindent By Theorem \ref{thm:uipkt} and Propositions \ref{prop:mirror} and \ref{prop:uipmirror}, we can establish the following. 

\begin{corollary}
\label{cor:uipmirrorKt}
If $\mathbf{Kt}\oplus \Delta$ has the uniform interpolation property, then $\mathbf{Kt}\oplus \mathtt{m}[\Delta]$ has the uniform interpolation property. 
\end{corollary}

Recall that a formula $\varphi$ is constant if $\mathsf{P}(\varphi)$ $=$ $\varnothing$. 
The following extends Rautenberg's result on Craig interpolation~\cite[Theorem 3]{Rautenberg1983} and Kurahashi's results on uniform (Lyndon) interpolation~\cite[Proposition 3]{Kurahashi2020} for basic modal logic to the case of basic tense logic. 

\begin{proposition}
\label{prop:uipconstant}
Let $\Delta$ be a set of constant formulas. If $\Lambda$ has the uniform interpolation property, then $\Lambda \oplus \Delta$ also has the uniform interpolation property. 
\end{proposition}

\begin{proof}
Let $\Delta$ be a set of constant formulas and suppose that $\Lambda$ has the uniform interpolation property. 
To show that $\Lambda \oplus \Delta$ has the uniform interpolation property, let us fix any formula $\varphi$ and $p \in \mathsf{Prop}$. 
By supposition, there exists a uniform interpolant $\Exi{p}\varphi$ in $\Lambda$ such that it satisfies the specified three conditions. 
We show that the same $\Exi{p}\varphi$ is a (post-) uniform interpolant for $(\varphi,p)$ in $\Lambda \oplus \Delta$. 
The conditions $(\mathtt{variable})$ and $(\mathtt{derivability})$
are easy to verify, so we focus on the final condition of $(\mathtt{uniformity})$. 
Fix any formula $\psi$ such that $p \notin \mathsf{P}(\psi)$ ($p$-freeness) and $\varphi \to \psi \in \Lambda \oplus \Delta$. Our goal is to show that $\Exi{p}\varphi \to \psi \in \Lambda \oplus \Delta$. 
It follows from $\varphi \to \psi \in \Lambda \oplus \Delta$ that there exists a finite multiset $\Delta'$ constructed from a finite subset of $\Delta$ and a family $(\mathsf{L}_{\delta})_{\delta \in \Delta'} \subseteq \setof{\Box,\bbox}^{< \omega}$ such that $\bigwedge \inset{ \mathsf{L}_{\delta} \delta}{\delta \in \Delta'} \to (\varphi \to \psi) \in \Lambda$ (since $\Delta$ is a set of constant formulas). 
Put $\bigwedge\mathop{\mathsf{L}}\Delta'$ := $\bigwedge \inset{ \mathsf{L}_{\delta} \delta }{\delta \in \Delta'}$. 
Then $\varphi \to (\bigwedge\mathop{\mathsf{L}}\Delta' \to \psi) \in \Lambda$.
Since $\Delta$ is a set of constant formulas, $\bigwedge\mathop{\mathsf{L}}\Delta'$ is also constant. 
So $\bigwedge\mathop{\mathsf{L}}\Delta' \to \psi$ is $p$-free. By supposition, we get $\Exi{p}\varphi \to (\bigwedge\mathop{\mathsf{L}}\Delta'\to \psi) \in \Lambda \subseteq \Lambda \oplus \Delta$ hence $\bigwedge\mathop{\mathsf{L}}\Delta' \to (\Exi{p}\varphi \to \psi) \in  \Lambda \oplus \Delta$. 
Since 
$\bigwedge\mathop{\mathsf{L}}\Delta' \in \Lambda \oplus \Delta$, this enables us to conclude that 
$\Exi{p}\varphi \to \psi \in \Lambda \oplus \Delta$, as desired. 
\end{proof}

By Theorem \ref{thm:uipkt} and Proposition \ref{prop:uipconstant}, we can establish the following. 

\begin{corollary}
\label{cor:uipconstantkt}
Let $\Delta$ be a set of constant formulas. Then $\mathbf{Kt}\oplus \Delta$ has the uniform interpolation property. 
Therefore, for every $\Delta \subseteq \setof{\mathtt{LM}, \mathtt{FM}, \mathtt{D}_{\Box}, \mathtt{D}_{\bbox}}$ of Table \ref{table:framedef}, $\mathbf{Kt} \oplus \Delta$ has the uniform interpolation property. 
\end{corollary}


\begin{definition}
We define the {\em boxdot translation} $(\cdot)^{\boxdot}: \mathsf{Form} \to \mathsf{Form}$ as follows:
\[
\begin{array}{rllrlllrll}
p^{\boxdot}  &:=&p, &
\bot^{\boxdot}  &:=&\bot, &
(\varphi \to \psi)^{\boxdot} &:=& \varphi^{\boxdot} \to \psi^{\boxdot}, \\
(\Box \varphi)^{\boxdot} &:=& \varphi^{\boxdot} \land \Box (\varphi^{\boxdot}),&
(\bbox \varphi)^{\boxdot} &:=& \varphi^{\boxdot} \land \bbox (\varphi^{\boxdot}).& && \\
\end{array}
\]
Given any normal tense logic $\Lambda$, we define ${\boxdot}^{-1}(\Lambda)$ := $\inset{\varphi}{\varphi^{\boxdot} \in \Lambda}$. 
\end{definition}


\begin{proposition}
\label{prop:invboxdot}
Let $\Lambda$ be a normal tense logic. 
\begin{enumerate*}[label=\arabic*., ref=\arabic*]
    \item \label{item:equiv}If $\Box p \to p \in \Lambda$, then $\varphi \leftrightarrow \varphi^{\boxdot} \in \Lambda$. 
    \item \label{item:invntl}${\boxdot}^{-1}(\Lambda)$ is a normal tense logic.
    \item \label{item:invntlref} $\Box p \to p \in {\boxdot}^{-1}(\Lambda)$. 
    \item \label{item:invntlbelowt} ${\boxdot}^{-1}(\Lambda) \subseteq \Lambda \oplus \setof{\Box p \to p}$. 
    \item \label{item:invequivt}
    If $\Lambda \subseteq \boxdot^{-1}(\Lambda)$ then ${\boxdot}^{-1}(\Lambda) = \Lambda \oplus \setof{\Box p \to p}$.
\end{enumerate*}
\end{proposition}

\begin{proof}
\begin{enumerate}
    \item We proceed by induction on $\varphi$. We only comment on the case where $\varphi$ is of the form $\bbox \psi$. We establish $\bbox \psi \leftrightarrow (\psi^{\boxdot} \land \bbox (\psi^{\boxdot}))$. By the induction hypothesis, it suffices to establish that $\bbox \psi \to \psi \in \Lambda$, which is obtained from item \ref{item:mirrorref} of Proposition \ref{prop:mirror}.  
    \item We only check axiom $(\Box \bdia)$, i.e., $p \to \Box \bdia p \in \boxdot^{-1}(\Lambda)$. 
    We show that $(p \to \Box \bdia p)^{\boxdot} \in \Lambda$. 
    We note that $(p \to \Box \bdia p)^{\boxdot}$ $\leftrightarrow (p \to (p \lor \bdia p)  \land  \Box (p \lor \bdia p)) \in \Lambda$. 
    But $(p \to (p \lor \bdia p)  \land  \Box (p \lor \bdia p)) \in \Lambda$ by axiom $(\Box \bdia)$.    
    \item It suffices to prove that $(\Box p \to p)^{\boxdot} \in \Lambda$. But $(\Box p \to p)^{\boxdot}$ := $(\Box p \land p) \to p$ is an instance of tautologies, which is in $\Lambda$.    
    \item Fix any formula $\varphi$ such that $\varphi^{\boxdot} \in \Lambda$. We show that $\varphi \in \Lambda \oplus \setof{\Box p \to p}$.  
    By item \ref{item:equiv}, we obtain $\varphi \leftrightarrow \varphi^{\boxdot} \in \Lambda \oplus \setof{\Box p \to p}$. 
    So it suffices to show that $\varphi^{\boxdot} \in \Lambda \oplus \setof{\Box p \to p}$. 
    But this is clear from $\varphi^{\boxdot} \in \Lambda$. 
    \item Suppose that $\Lambda \subseteq \boxdot^{-1}(\Lambda)$. 
    By item \ref{item:invntlref} and the supposition, we obtain $\Lambda \oplus \setof{\Box p \to p}\subseteq \boxdot^{-1}(\Lambda)$. 
    It follows from item \ref{item:invntlbelowt} that $\boxdot^{-1}(\Lambda) = \Lambda \oplus \setof{\Box p \to p}$, as desired. 
    \qedhere
\end{enumerate}
\end{proof}

The following provides an extension of Kurahashi's result~\cite[Proposition 3]{Kurahashi2020} on the boxdot translation from normal modal logics to normal tense logics.

\begin{proposition}
\label{prop:refuip}
Let $\Lambda$ be a normal tense logic such that $\Lambda \subseteq \boxdot^{-1}(\Lambda)$. 
If $\Lambda$ has the uniform interpolation property, then $\boxdot^{-1}(\Lambda)$ $(= \Lambda \oplus \setof{\Box p \to p})$ also has the uniform interpolation property. 
\end{proposition}

\begin{proof}
Let $\Lambda$ be a normal tense logic such that $\Lambda \subseteq \boxdot^{-1}(\Lambda)$ and that $\Lambda$ has the uniform interpolation property. 
By item \ref{item:invequivt} of Proposition \ref{prop:invboxdot}, we obtain $\boxdot^{-1}(\Lambda) = \Lambda \oplus \setof{\Box p \to p}$. 
For the uniform interpolation property of $\boxdot^{-1}(\Lambda)$, let us fix any $(\varphi, p) \in \mathsf{Form} \times \mathsf{Prop}$.
Let $\Exi{p}\varphi^{\boxdot}$ be a uniform interpolant for $(\varphi^{\boxdot},p)$ in $\Lambda$. We show that $\Exi{p}\varphi^{\boxdot}$ is a uniform interpolant for $(\varphi,p)$ in $\boxdot^{-1}(\Lambda)$. Since the boxdot translation does not change the set of propositional variables in a formula, $\mathsf{P}(\Exi{p}\varphi^{\boxdot}) \subseteq \mathsf{P}
(\varphi^{\boxdot}) \setminus \setof{p}$ = $\mathsf{P}
(\varphi) \setminus \setof{p}$.
So $(\mathtt{variable})$ condition holds. 
For $(\mathtt{derivability})$, 
we prove that $\varphi \to \Exi{p}\varphi^{\boxdot} \in \boxdot^{-1}(\Lambda)$. 
Since $\varphi^{\boxdot} \to \Exi{p}\varphi^{\boxdot} \in \Lambda$ and $\varphi \leftrightarrow \varphi^{\boxdot} \in \boxdot^{-1}(\Lambda)$ (by items \ref{item:equiv}, \ref{item:invntl} and \ref{item:invntlref} of Proposition \ref{prop:invboxdot}), $\Lambda \subseteq \boxdot^{-1}(\Lambda)$ implies that 
$\varphi \to \Exi{p}\varphi^{\boxdot} \in \boxdot^{-1}(\Lambda)$, as required. 
For $(\mathtt{uniformity})$, fix any $p$-free formula $\psi$ with $\varphi \to \psi \in  \boxdot^{-1}(\Lambda)$. 
We show that $\Exi{p}\varphi^{\boxdot} \to \psi \in  \boxdot^{-1}(\Lambda)$.
By $\varphi \to \psi \in  \boxdot^{-1}(\Lambda)$, we obtain $\varphi^{\boxdot} \to \psi^{\boxdot} \in \Lambda$ by definition. 
Since $\psi^{\boxdot}$ is $p$-free, we obtain $\Exi{p}\varphi^{\boxdot} \to \psi^{\boxdot} \in \Lambda$. 
By $\Lambda \subseteq \boxdot^{-1}(\Lambda)$, it holds that $\Exi{p}\varphi^{\boxdot} \to \psi^{\boxdot} \in 
\boxdot^{-1}(\Lambda)$. 
Since $\rho \leftrightarrow \rho^{\boxdot} \in \boxdot^{-1}(\Lambda)$ for all formulas $\rho$, we obtain $\psi \leftrightarrow \psi^{\boxdot} \in  \boxdot^{-1}(\Lambda)$. 
We conclude that $\Exi{p}\varphi^{\boxdot} \to \psi \in 
\boxdot^{-1}(\Lambda)$, as desired.
\end{proof}

By Proposition \ref{prop:refuip}, we obtain the following specific results on normal tense logics. 

\begin{corollary}
\label{cor:boxdot}
\begin{enumerate}
    \item $\mathbf{Kt} \oplus \mathtt{T}_{\Box}$ has the uniform interpolation property. 
    \item Let $\Delta \subseteq \setof{\mathtt{D}_{\Box}, \mathtt{D}_{\bbox}}$. 
    If $\mathbf{Kt}\oplus \mathtt{4}_{\Box} \oplus \Delta$ has the uniform interpolation property, then $\mathbf{Kt}\oplus \setof{\mathtt{4}_{\Box}, \mathtt{T}_{\Box}}$ also has the uniform interpolation property. 
    \end{enumerate}
\end{corollary}

\section{Failure of Uniform Interpolation Property of Tense $\mathbf{S4}$}
\label{sec:uipfailureS4t}

This section demonstrates that the tense expansion of modal logic $\mathbf{S4}$ lacks the uniform interpolation property. To establish this, we first characterize the uniform interpolation property for the normal tense logic of a frame class in terms of the bisimulation expansion property, under the assumption that the logic of the frame class is compact and the frame class satisfies the following Hennessy-Milner property. 

\begin{definition}
A frame class $\mathbb{F}$ has the \emph{Hennessy-Milner property} if, for every two pointed models $(M_{1},w_{1})$ and $(M_{2}, w_{2})$ from $\mathbb{F}$, $(M_{1},w_{1})  \moeq^{\mathsf{P}} (M_{2},w_{2})$ implies 
$(M_{1},w_{1})  \bisim^{\mathsf{P}} (M_{2},w_{2})$ for every $\mathsf{P} \subseteq \mathsf{Prop}$. 
\end{definition}

It is easy to see that a class of finite frames always has the Hennessy-Milner property. 

\begin{definition}
\label{dfn:compact}
Given a class of frames $\mathbb{F}$, we say that the normal tense logic 
$\Lambda_{\mathbb{F}}$ is \emph{compact} if, whenever every finite subset 
$\Gamma' \subseteq \Gamma$ is satisfiable in $\mathbb{F}$, the set $\Gamma$ 
itself is also satisfiable in $\mathbb{F}$.
\end{definition}


\begin{proposition}
\label{prop:uip2bep}
Let $\mathbb{F}$ be a frame class satisfying the Hennessy-Milner property such that $\Lambda_{\mathbb{F}}$ is compact. The frame class $\Lambda_{\mathbb{F}}$ has the uniform interpolation property iff $\mathbb{F}$ has the bisimulation expansion property. 
\end{proposition}

\begin{proof}
    Let $\mathbb{F}$ be a frame class such that it has the Hennessy-Milner property. 
    The right-to-left direction is already established in Lemma~\ref{lem:bep2uip}. 
    So, we establish the left-to-right direction. 
    Assume that $\Lambda_{\mathbb{F}}$ has the uniform interpolation property. 
    To establish the bisimulation expansion property of $\mathbb{F}$, 
    let us fix any subsets $\mathsf{P}_{1}, \mathsf{P}_{2} \subseteq \mathsf{Prop}$ such that $\mathsf{P}_{1}$ is finite, any pointed models $(M_{1}, w_{1})$ and $(M_{2},w_{2})$ from $\mathbb{F}$ and any $n \in \omega$. 
    Suppose that $(M_{1},w_{1}) \moeq_{n}^{\mathsf{P}_{1} \cap \mathsf{P}_{2}} (M_
    {2},w_{2})$.  
    Put $\mathsf{P}_{1} \setminus \mathsf{P}_{2}$ := $\setof{p_{1},\ldots,p_{n}}$.
    Recall $\chi_{(M_{1},w_{1})}^{\mathsf{P}_{1};n}$ := $\bigwedge \inset{\chi \in \Phi (\mathsf{P}_{1};n)}{M_{1},w_{1} \models \chi}$. 
    By our supposition, we can construct a uniform interpolant 
    $\theta$ := $\Exi{p_{1}}\cdots \Exi{p_{n}} \chi_{(M_{1},w_{1})}^{\mathsf{P}_{1};n}$.     
    We have $\mathsf{P}\left(\theta \right)$ = $\mathsf{P} \left(\chi_{(M_{1},w_{1})}^{\mathsf{P}_{1};n}\right) \setminus \setof{p_{1},\ldots,p_{n}}$ = $\mathsf{P}_{1} \setminus (\mathsf{P}_{1} \setminus \mathsf{P}_{2})$ = $\mathsf{P}_{1} \cap \mathsf{P}_{2}$. 
    It follows from $\mathbb{F} \models \chi_{(M_{1},w_{1})}^{\mathsf{P}_{1};n} \to \theta$, 
    $(M_{1},w_{1}) \models \chi_{(M_{1},w_{1})}^{\mathsf{P}_{1};n}$ and
    $M_{1},w_{1} \moeq_{n}^{\mathsf{P}_{1} \cap \mathsf{P}_{2}} (M_
    {2},w_{2})$ that 
    $M_{2},w_{2} \models \theta$. 
    Define $\mathsf{Th}(M_
    {2}, w_{2})$ := $\inset{\rho \in \mathsf{Form}(\mathsf{P}_{2})}{M_{2},w_{2}\models \rho}$. 
    Then $\setof{\theta} \cup \mathsf{Th}(M_
    {2}, w_{2})$ is satisfiable in $\mathbb{F}$. 
    We can also establish that 
    $\setof{ \chi_{(M_{1},w_{1})}^{\mathsf{P}_{1};n} } \cup \mathsf{Th}(M_
    {2}, w_{2})$ is satisfiable in $\mathbb{F}$.
    Suppose otherwise. 
    Since $\mathsf{Th}(M_
    {2}, w_{2})$ is closed under conjunctions and $\Lambda_{\mathbb{F}}$ is compact, we can find a formula $\rho \in \mathsf{Th}(M_
    {2}, w_{2})$ such that  $\mathbb{F} \models \chi_{(M_{1},w_{1})}^{\mathsf{P}_{1};n} \to \neg\rho$ holds. 
    By $\mathsf{P} (\rho) \cap \setof{p_{1},\ldots,p_{n}}$ = $\varnothing$, 
    $\rho$ is $\setof{p_{1},\ldots,p_{n}}$-free. 
    So, we obtain 
    $\mathbb{F} \models \theta \to \neg \rho$, which means that $\setof{\theta} \cup \mathsf{Th}(M_{2},w_{2})$ is not satisfiable in $\mathbb{F}$. 
    However, $\setof{\theta} \cup \mathsf{Th}(M_{2},w_{2})$ is satisfiable in $\mathbb{F}$. We obtain a desired contradiction. 
    By the satisfiability of 
        $\setof{ \chi_{(M_{1},w_{1})}^{\mathsf{P}_{1};n} } \cup \mathsf{Th}(M_
    {2}, w_{2})$ in $\mathbb{F}$, we can find a pointed model $(N,x)$ from $\mathbb{F}$ such that $N,x \models \chi_{(M_{1},w_{1})}^{\mathsf{P}_{1};n}$ and $N,x \models \mathsf{Th}(M_
    {2}, w_{2})$. 
    It follows that $(M_{1},w_{1}) \bisim_{n}^{\mathsf{P}_{1}} (N,x)$ and $(N,x) 
    \moeq^{\mathsf{P}_{2}} (M_{2},w_{2})$ hence $(N,x) 
    \bisim^{\mathsf{P}_{2}} (M_{2},w_{2})$ by the Hennessy-Milner property of $\mathbb{F}$. 
    This finishes showing the bisimulation expansion property of $\mathbb{F}$. 
\end{proof}

By Lemma \ref{lem:bep2uip}, Corollary \ref{cor:bisimquan} and Proposition \ref{prop:uip2bep}, we can also obtain the following. 

\begin{corollary}
\label{cor:hmuip2bisimquan}
Let $\mathbb{F}$ be a frame class satisfying the Hennessy-Milner property such that $\Lambda_{\mathbb{F}}$ is compact. Then if $\Lambda_{\mathbb{F}}$ has the uniform interpolation property, the uniform interpolant $\Exi{p}\varphi$ for $(\varphi,p) \in \mathsf{Form} \times \mathsf{Prop}$ in $\Lambda_{\mathbb{F}}$ is a bisimulation quantifier in the sense of Corollary \ref{cor:bisimquan}.
\end{corollary}

\begin{definition}
Let $\mathbf{S4t}$ be the smallest normal tense logic that contains $\mathtt{T}_{\Box}$ and $\mathtt{4}_{\Box}$ in Table \ref{table:framedef}. Let $\mathbb{F}^{\mathtt{fin}}_{\mathbf{S4}}$ be the class of all reflexive and transitive finite frames $(W,R)$, i.e., all quasi-orders $(W,R)$. 
\end{definition}

\noindent The following can be readily shown via the canonical model construction and 
filtration technique. The finite model property of $\mathbf{S4t}$ was also established in~\cite[Theorem 3.6]{Sano2020a} using an analytic sequent calculus.  

\begin{proposition}
\label{prop:S4t}
The normal tense logic $\mathbf{S4t}$ is sound and complete for 
the class $\mathbb{F}^{\mathtt{fin}}_{\mathbf{S4}}$ of all finite reflexive and transitive frames, i.e., $\mathbf{S4t}$ = $\Lambda_{\mathbb{F}^{\mathtt{fin}}_{\mathbf{S4}}}$. 
Moreover, $\mathbf{S4t}$ = $\Lambda_{\mathbb{F}^{\mathtt{fin}}_{\mathbf{S4}}}$ is compact. 
\end{proposition}

The lemma below extends Visser's theorem for the modal logic $\mathbf{S4}$~\cite[Theorem 16.3]{Visser1996preprint} to a tense-logic setting. The inclusion of the additional conditions $(\mathtt{Forth}^{-1})$ and $(\mathtt{Back}^{-1})$ does not alter the original line of reasoning. 
Given a model $M$ := $(W,R,V)$. we use $V[s:= X]$ to denote the same valuation as $V$ except that it sends a propositional variable $s$ to a set $X \subseteq W$, and we define $M[s:=X] := (W,R,V[s:=X])$.

\begin{lemma}
\label{lem:bqaltS4t}
Define $\varphi(p,q,r)$ as 
$\varphi(p,q,r) := p \land \Box (p \to \dia q) \land \Box(q \to \dia p) \land \Box (p \to r) \land \Box (q \to \neg r)$. Suppose that $\mathbf{S4t}$ has the uniform interpolation property. 
Let $(M,w)$ be a pointed model from $\mathbb{F}^{\mathtt{fin}}_{\mathbf{S4}}$, where $M$ $=$ $(W,R,V)$. Then $M,w \models \Exi{p}\Exi{q} \varphi (p,q,r)$ iff there exists an infinite $R$-sequence $(w_{n})_{n \in \omega} \subseteq W$ such that $w_{0}$ = $w$, 
$w_{k}Rw_{k+1}$, $M,w_{2k} \models r$ and $M,w_{2k + 1} \not\models r$ hold, for all $k \in \omega$.  
\end{lemma}

\begin{proof}
Suppose that $\mathbf{S4t}$ has the uniform interpolation property. For the right-to-left direction, assume there exists an infinite $R$-sequence $(w_{n})_{n \in \omega} \subseteq W$ satisfying the stated conditions. 
Define $E := \inset{w_{2k}}{k \in \omega}$ and $O := \inset{w_{2k+1}}{k \in \omega}$. By $r \notin \setof{p,q}$, we have $M[p:=E][q:= O],w \models \varphi(p,q,r)$. Moreover, the following hold:
\[(M,w) \bisim^{\mathsf{Prop}\setminus\setof{p}} (M[p:= E],w) \bisim^{\mathsf{Prop}\setminus\setof{q}} (M[p:=E][q:= O],w).
\]
Therefore, we obtain $M,w \models \Exi{p}\Exi{q} \varphi (p,q,r)$ by Proposition \ref{prop:S4t} and Corollary \ref{cor:hmuip2bisimquan}.

For the left-to-right direction, suppose $M,w \models \Exi{p}\Exi{q} \varphi (p,q,r)$. By Proposition \ref{prop:S4t} and Corollary \ref{cor:hmuip2bisimquan}, there exist two pointed models $(M_{1},w_{1})$ and $(M_{2},w_{2})$ such that 
$(M,w) \bisim^{\mathsf{Prop}\setminus\setof{p}} (M_{1},w_{1}) \bisim^{\mathsf{Prop}\setminus\setof{q}} (M_{2},w_{2})$ and $M_{2},w_{2} \models \varphi (p,q,r)$. The latter implies that there exists an infinite $R_{2}$-path in $M_{2}$ along which the truth value of $r$ alternates at every step. By the transitivity of bisimulation, it follows that $(M,w) \bisim^{\mathsf{Prop}\setminus\setof{p,q}} (M_{2},w_{2})$. Since $r \notin \setof{p,q}$, the valuation of $r$ is preserved between $M$ and $M_{2}$. Therefore, by the $(\mathtt{Back})$ condition, we can construct a corresponding infinite $R$-path in $M$ along which $r$ alternates between being true and false at every step, as desired.
\end{proof}

\begin{figure}[htbp]
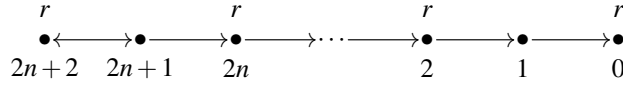

    \centering
\vspace{-0.5cm}
    \begin{small}
\[
\xygraph{!~:{@{<->}}
    \bullet ([]!{+(0,-.3)} {2n+2}) ([]!{+(0,0.3)} {r}) : [r]
    \bullet ([]!{+(0,-.3)} {2n+1}) -@{->} [r] 
    \bullet ([]!{+(0,-.3)} {2n}) 
    ([]!{+(0,0.3)} {r}) -@{->} [r]  \cdots  -@{->} [r]
    \bullet ([]!{+(0,-.3)} {2})([]!{+(0,0.3)} {r}) -@{->} [r]
    \bullet ([]!{+(0,-.3)} {1}) -@{->} [r]
    \bullet ([]!{+(0,-.3)} {0})
    ([]!{+(0,0.3)} {r})
    }
\]
\end{small}
\vspace{-0.5cm}
    \caption{A model $M$ in the proof of Lemma \ref{lem:undefalt}, with reflexive loops and transitive compositions omitted.}
    \label{fig:model2n+2}
\end{figure}

The following result extends Ghilardi and Zawadowski's~\cite{Ghilardi1995} theorem on the failure of the uniform interpolation property in $\mathbf{S4}$ to a tense-logical setting. Our proof builds upon the one given in~\cite[Theorem 16.4]{Visser1996preprint}, which provides an alternative proof of their theorem. The model used in our proof (see Figure~\ref{fig:model2n+2}) is identical to the one presented in~\cite[p.~269]{Ghilardi1995} and~\cite[Theorem 16.4]{Visser1996preprint}.
Furthermore, the addition of the conditions $(\mathtt{Forth}^{-1})$ and $(\mathtt{Back}^{-1})$ to the notion of (non-temporal) layered bisimulations (e.g., from~\cite{Visser1996}) remains consistent with the model construction in~\cite{Visser1996preprint}.

\begin{lemma}
\label{lem:undefalt}
There is no formula $\psi(r) \in \mathsf{Form}(\setof{r})$ such that, for all pointed models $(M,w)$ from $\mathbb{F}^{\mathtt{fin}}_{\mathbf{S4}}$ such that $M$ = $(W,R,V)$, the following equivalence holds: $M,w \models \psi(r)$ iff there exists an infinite $R$-sequence $(w_{n})_{n \in \omega} \subseteq W$ such that $w_{0}$ = $w$, $w_{k}Rw_{k+1}$, $M,w_{2k} \models r$ and $M,w_{2k + 1} \not\models r$ hold, for all $k \in \omega$. 
\end{lemma}

\begin{proof}
Suppose for contradiction that we can find a formula $\psi(r)$ such that the stated equivalence holds for all pointed models $(M,w)$ from $\mathbb{F}^{\mathtt{fin}}_{\mathbf{S4}}$. Fix some $n \in \omega$ such that $\mathtt{d}(\psi(r)) \leqslant 2n$. 
Define a model $M = (W,R,V)$ as follows (see Figure \ref{fig:model2n+2}): $W := \{0, 1, \ldots, 2n, 2n+1, 2n+2\}$, $xRy$ iff $x \geqslant y$ or $x \in \setof{2n+1, 2n+2}$; $x \in V(r)$ iff $x$ is even. 
The set $\setof{2n+1,2n+2}$ forms a \emph{cluster} in terms of $R$. 
It is noted that 
$M,2n+2 \models \psi(r)$ but $M,2n \not\models \psi(r)$. 
Let us define $Z \subseteq W \times \omega \times W$ by 
\begin{center}
$(x, n, x') \in Z$ iff $x=x'$ or ($x \equiv x'\pmod 2$ and $n \leqslant \min (x,x')$). 
\end{center}
By definition, we obtain $(2n, 2n, 2n+2) \in Z$. 

\begin{clma1}
\label{clm:S4ce}
$Z$ is a layered (temporal) $\setof{r}$-bisimulation between $M$ and $M$. 
\end{clma1}

\begin{pfclma}
It is immediate to see that $(\mathtt{Atom})$ holds by definition. In what follows, we check the remaining conditions. Suppose that $(x,k+1,x') \in Z$. 
If $x=x'$, then all four required conditions hold immediately. For example, 
$(\mathtt{Forth})$ holds as follows: Assume $xRy$. Then $(y,k,y) \in Z$ holds by definition. So, we can find $y' \in W$ such that $x$ = $x'Ry'$ and $(y,k,y') \in Z$, as desired. 

Let us suppose that $x \neq x'$. Since $(x,k+1,x') \in Z$, it holds that ($x \equiv x'\pmod 2$ and $k + 1 \leqslant \min (x,x')$). 
Now we assume without loss of generality that $x < x'$. 
Since our definition of $Z$ is symmetric with respect to $x$ and $x'$, it suffices to prove the four remaining conditions under this assumption $x<x'$. 
First, we establish that $x \notin \setof{2n+1,2n+2}$. 
Suppose otherwise, i.e., 
$x \in \setof{2n+1,2n+2}$. 
If $x$ = $2n+1$, then $x'$ = $2n+2$, which is a contradiction with $x \equiv x'\pmod 2$. 
If $x$ = $2n+2$, then it is impossible to have $x<x'$. 
So, we have shown that $x \notin \setof{2n+1,2n+2}$ hence $x \leqslant 2n$. 

It is noted that arguments below for $(\mathtt{Forth})$ and $(\mathtt{Back})$ are the same as those in the proof of~\cite[Theorem 16.4]{Visser1996preprint}. 

\begin{itemize}
    \item \fbox{$\mathtt{Forth}$}
    Suppose that $xRy$, i.e., $y \leqslant x$ by $x \notin \setof{2n+1,2n+2}$. 
    We show that there exists $y \in W$ such that $x'Ry'$ and $(y,k,y') \in Z$. We define $y' := y$. Since $y \leqslant x < x'$, we obtain $x'Ry$. We get $(y,k,y) \in Z$ immediately by definition. 
    \item \fbox{$\mathtt{Back}$}
    Suppose that $x'Ry'$. We show that there exists $y \in W$ such that $xRy$ and $(y,k,y') \in Z$. 
    We divide our argument into the following two cases: i) $y' \leqslant x$ and ii) $y' > x$. 
    \begin{itemize}
    \item For case i) of $y' \leqslant x$, we proceed as follows: 
    Define $y := y'$. 
    It follows from $y'\leqslant x$ that $xRy'$ by definition. Moreover, $(y',k,y') \in Z$ holds immediately. 
    \item For case ii) of $y' > x$, we define our candidate $y$ as follows: 
    \[
    y:= 
    \begin{cases}
    x    &\text{if $y' \equiv x\pmod 2$};\\
    x-1   &\text{if $y' \not\equiv x\pmod 2$}.
    \end{cases}
    \]
    It is noted that $k+1 
\leqslant \min (x,x') = x$ implies $x \geqslant 1$ and so $x-1$ is well-defined. 
    By definition, we obtain $xRy$. 
    For both cases of $y$, it suffices to show that $(y,k,y') \in Z$. We note from $y' > x$ that $y' > y$ hence $y' \neq y$. 
    It is easy to see that $y \equiv y' \pmod 2$. So, it remains to establish $k \leqslant \min (y,y')$. This is verified as follows: 
    Since $k+1 \leqslant \min (x,x') = x$ (by $x <x'$), we obtain $k \leqslant x-1 \leqslant y = \min (y,y')$ (by $y' > y$), as desired. 
    \end{itemize} 
    \item \fbox{$\mathtt{Forth}^{-1}$}
    Suppose that $yRx$. 
    Whether $y \in \setof{2n+1, 2n+2}$ holds or not, we obtain $x \leqslant y$ (by $x \notin  \setof{2n+1, 2n+2}$). 
    We show that there exists $y' \in W$ such that $y'Rx'$ and $(y,k,y') \in Z$.  
    We divide our argument into the following two cases: i) $x' \leqslant y$ and ii) $x' > y$. 
    \begin{itemize}
        \item i) Assume $x' \leqslant y$. 
        Then $yRx'$. Together with $(y,k,y)\in Z$, we have obtained our goal. 
        \item ii) Assume $x' > y$. 
        Let us define $\mathtt{suc}: W \to W$ by $\mathtt{suc}(z)$ 
        $:=$ $z+1$ (if $z \leqslant 2n+1$); $\mathtt{suc}(z)$ 
        $:=$ $2n+1$ (if $z = 2n+2$).  
        We also define $y'$ as follows: 
        \[
        y' :=
        \begin{cases}
        x'   & \text{if $y \equiv x' \pmod 2$}; \\
        \mathtt{suc}(x')    & \text{if $y \not\equiv x' \pmod 2$}.
        \end{cases}
        \]
        For both choices of $y'$, we obtain $y'Rx'$.  
        By definition, it is clear that $y \equiv y' \pmod 2$. 
        \begin{itemize}
            \item Let $y \equiv x' \pmod 2$. 
            To show that $(y,k,y') \in Z$, it suffices to establish $k \leqslant \min (y,y')$ = $y$ as follows: $k \leqslant k+1 \leqslant \min (x,x')$ = $x \leqslant y$, as desired.  
            \item Let $y \not\equiv x' \pmod 2$. 
            If $y \in \setof{2n+1,2n+2}$, 
            then $(x',y)$ = $(2n+2,2n+1)$ by $x'>y$. It follows that $y'$ = $\mathtt{suc}(x')$ = $2n+1$ = $y$ hence $(y,k,y') \in Z$. 
            In what follows, we let 
            $y \notin \setof{2n+1,2n+2}$, i.e., $y \leqslant 2n$.   
            By $x' > y$, we get $y' \neq y$ by definition of $y'$ (this also  holds when $x'$ = $2n+2$ and so $y'$ = $2n+1$). 
            To show that $(y,k,y') \in Z$, it suffices to establish $k \leqslant \min (y,y')$ = $y$ as follows: 
            $k \leqslant k+1 \leqslant \min (x,x')$ $=$ $x \leqslant y$, as desired.  
        \end{itemize}
    \end{itemize}
    
    \item \fbox{$\mathtt{Back}^{-1}$}
    Suppose that $y'Rx'$. We show that there exists $y \in W$ such that $yRx$ and $(y,k,y') \in Z$. 
    Assume $y' \in \setof{2n+1,2n+2}$. 
    Then we obtain $y' Rx$ and $(y',k,y') \in Z$. So it suffices to define $y := y'$. 
    Assume otherwise, i.e., $y' \notin \setof{2n+1,2n+2}$. Then it follows from $y'Rx'$ that $x' \leqslant y'$ by definition of $R$.    
    By $x<x'$, we obtain $x < y'$ hence $y'Rx$. 
    Since $(y',k,y') \in Z$ holds immediately, it suffices to define $y := y'$. 
    \end{itemize}
This finishes establishing that 
$Z$ is a layered $\setof{r}$-bisimulation between $M$ and $M$.  
\end{pfclma}

By Claim \ref{clm:S4ce} and Proposition \ref{prop:nbismnmoeq}, it follows from $M,2n+2 \models \varphi(r)$ that 
$M,2n \models \varphi(r)$, 
which is a contradiction with $M,2n \not\models \varphi(r)$.
\end{proof}

\begin{theorem}
\label{thm:failuareuipS4t}
The normal tense logic $\mathbf{S4t}$ does not have the uniform interpolation property. 
\end{theorem}

\begin{proof}
Suppose for contradiction that 
$\mathbf{S4t}$ has the uniform interpolation property. 
For the formula $\varphi(p,q,r)$ in the statement of Lemma \ref{lem:bqaltS4t} and $p,q \in \mathsf{Prop}$, we can construct the formula $\Exi{p}\Exi{q} \varphi(p,q,r)$ in terms of uniform interpolants, whose semantic clause at $(M,w)$ from $\mathbb{F}^{\mathtt{fin}}_{\mathbf{S4}}$ means that 
there exists an infinite $R$-path along which the propositional variable $r$ alternates between being true and false at every step. 
However, by noting that $\mathsf{P}(\Exi{p}\Exi{q} \varphi(p,q,r)) \subseteq \setof{r}$, this is impossible by Lemma \ref{lem:undefalt}. 
\end{proof}

By Corollary \ref{cor:boxdot} and Theorem \ref{thm:failuareuipS4t}, we can also derive the following. 

\begin{corollary}
\label{cor:failureK4t}
For every $\Delta \subseteq \setof{\mathtt{D}_{\Box}, \mathtt{D}_{\bbox}}$,  $\mathbf{Kt}\oplus \mathtt{4}_{\Box} \oplus \Delta$ lacks the uniform interpolation property. 
\end{corollary}

\section{Further Directions}
\label{sec:concl}

We have established that the basic tense logic $\mathbf{Kt}$ has the uniform interpolation property. It is straightforward to extend our argument to the multimodal (or indexed) version of $\mathbf{Kt}$ with modalities $[a]$ and $[a^{\smallsmile}]$ ($a \in \mathsf{Mod}$), as the addition of indices does not alter the outline of the argument. This extension also allows us to establish the uniform interpolation property for star-free propositional dynamic logic with converse \cite{Harel2000}, as there exists a translation from this logic into the indexed version of $\mathbf{Kt}$ that preserves the set of propositional variables occurring in a formula.\footnote{The author would like to thank Yde Venema for this observation.}

There are several directions for future research. First, an intriguing open question arising from this work concerns the extension of nabla ($\nabla$) normal forms to a tense-logical setting (see also Appendix \ref{sec:canfor}). In basic modal logic, $\nabla$-based normal forms (or disjunctive forms characterized by non-interacting conjunctive components) share deep structural links with both uniform interpolation and bisimulation quantifiers (cf.~\cite[Section 5]{Bezhanishvili2025}). A natural next step is to determine whether an analogous normal form can be formulated for $\mathbf{Kt}$ to prove its uniform interpolation property, which remains an open problem in~\cite[Section 5]{Bezhanishvili2025}.

Second, we may investigate whether a Pitts--B\'{i}lkov\'{a}-style proof-theoretic approach (cf.~\cite{Pitts1992,Bilkova2007}) can be applied to the uniform interpolation property of $\mathbf{Kt}$ (even for modal logic $\mathbf{B}$). 

Third, we could provide further examples of normal tense logics with the uniform interpolation property. For example, we have not considered the axiom $\mathtt{5}_{\Box}$ ($\dia p \to \Box \dia p$) and its mirror image $\mathtt{5}_{\bbox}$ ($\bdia p \to \bbox \bdia p$) in this paper. In the context of basic modal logics, it is well known, e.g., in~\cite{Zakharyaschev2001}, that every normal modal logic containing $\mathtt{5}_{\Box}$ is \emph{locally tabular} (i.e., the number of non-equivalent formulas generated from a finite set of propositional variables is finite, cf.~\cite[p.19]{CZ1997}). 
Since local tabularity and the Craig interpolation property jointly imply the uniform interpolation property, it would be interesting to investigate the uniform interpolation property of the normal tense extension containing $\setof{\mathtt{5}_{\Box}, \mathtt{5}_{\bbox}}$.

Fourth, we may extend our argument to establish the uniform Lyndon interpolation property of $\mathbf{Kt}$, which further requires considering the polarities of propositional variables. As demonstrated by Kurahashi~\cite{Kurahashi2020}, this can be accomplished by introducing a polarized variant of layered bisimulation. 

Finally, basic tense logic was pioneered by Arthur Prior~\cite{Prior1957,Prior1967,Prior1968,Prior2003}, who subsequently extended $\mathbf{Kt}$ into a \emph{hybrid} logic equipped with nominals, the global modality, and a universal quantifier over nominals (cf.~\cite{Blackburn2020,Blackburn2023}). In light of recent studies on propositionally quantified hybrid logic, exploring the philosophical implications of the uniform interpolation property for $\mathbf{Kt}$ becomes pertinent, given that it enables the definition of a bisimulation quantifier~\cite{DAgostino2006,D’Agostino2015}, which is a distinct form of propositional quantification. 

\section*{Acknowledgements}

The author would like to thank the three reviewers for their constructive comments and suggestions, which have helped improve the quality of this submission. He would also like to thank Hans van Ditmarsch and Yde Venema for their insightful discussions on the uniform interpolation property in epistemic logic with distributed knowledge in Toulouse in early April 2025 and basic tense logic in Amsterdam in May 2025, respectively. This work was partially supported by JSPS KAKENHI Grants-in-Aid for Scientific Research (B) (Grant Number JP22H00597) and (C) (Grant Number JP25K03537).

\appendix

\section{Proof of Proposition \ref{prop:foreqfindep}}
\label{sec:canfor}

This section provides an outline of a proof of Proposition \ref{prop:foreqfindep} in terms of \emph{canonical formulas} defined in~\cite{Moss2007} (see also~\cite{Fang2019}). 

\begin{definition}
Given a finite set $\mathsf{P}\subseteq \mathsf{Prop}$ and 
$\mathsf{Q}\subseteq \mathsf{P}$ of propositional variables, 
we define an abbreviation 
$\pi (\mathsf{Q}, \mathsf{P})$ := $\bigwedge \mathsf{Q} \land \bigwedge \neg (\mathsf{P} \setminus \mathsf{Q})$. 
Given a finite set $\Psi$ of formulas, we also introduce the cover (or nabla) modalities $\nabla \Psi$ and $\bnabla \Psi$ as follows: 
\[
\begin{array}{llllll}
    \nabla \Psi &:=&  \bigwedge \dia \Psi \land \Box \bigvee \Psi, &
    \bnabla \Psi &:=& 
    \bigwedge \bdia \Psi \land \bbox \bigvee \Psi.
\end{array}
\]
\end{definition}
\noindent We note that $\Box \varphi$ and $\dia \varphi$ can be written as $\nabla \varnothing \lor \nabla \setof{\varphi}$ and $\nabla \setof{\top, \varphi}$, respectively. 

\begin{definition}
For a finite subset $\mathsf{P}\subseteq \mathsf{Prop}$ and $n \in \omega$, we define a \emph{finite} set $\mathsf{C}(\mathsf{P};n)$ of \emph{canonical formulas} with respect to $(\mathsf{P};n)$ as follows:
\[
\begin{array}{lll}
    \mathsf{C}(\mathsf{P};0)  &:=& 
    \inset{\pi(\mathsf{Q},\mathsf{P})}{\mathsf{Q} \subseteq \mathsf{P}}; \\
    \mathsf{C}(\mathsf{P};n+1) &:=& \inset{\pi \land \nabla \Psi \land \bnabla \Psi'}{\pi \in \mathsf{C}(\mathsf{P};0) \text{ and } \Psi, \Psi' \subseteq \mathsf{C}(\mathsf{P};n) }.\\ 
\end{array}
\]
\end{definition}
\noindent It is noted that $\mathsf{C}(\mathsf{P};n)$ is finite for each $n$ and so $\mathsf{C}(\mathsf{P};n)$ is well-defined. 

\begin{proposition}
\label{prop:canonicalformula}
Let $\mathsf{P}\subseteq \mathsf{Prop}$ be a finite subset and $n \in \omega$. 
\begin{enumerate}
    \item Given a pointed model $(M,w)$, there exists a unique $\chi \in \mathsf{C}(\mathsf{P};n)$ such that $M,w\models \chi$. 
    \item Let $\chi \in \mathsf{C}(\mathsf{P};n)$. For every $\varphi \in \mathsf{Form}(\mathsf{P};n)$, it holds that   $\mathbb{F}_{\mathtt{all}} \models \chi \to \varphi$ or $\mathbb{F}_{\mathtt{all}} \models \chi\to \neg \varphi$. 
    \item For all formulas $\varphi \in \mathsf{Form}(\mathsf{P};n)$, there exists ${\Phi} \subseteq \mathsf{C}(\mathsf{P};n)$ such that $\mathbb{F}_{\mathtt{all}} \models \varphi \leftrightarrow \bigvee \Phi$. 
\end{enumerate}
\end{proposition}

\begin{proof}
(Outline) Items 1 and 2 can be established similarly to~\cite[Lemma 2.6]{Moss2007} and~\cite[Lemma 2.7]{Moss2007}, respectively. Item 3 follows from item 2 (see also~\cite[Proposition 6]{Fang2019}). 
\end{proof}

\noindent Proposition \ref{prop:canonicalformula} (3) implies Proposition \ref{prop:foreqfindep}, i.e., the quotient set $\mathsf{Form}(\mathsf{P};n)/\equiv$ of $\mathsf{Form}(\mathsf{P};n)$ in terms of logical equivalence $\equiv$  is finite. 

\bibliographystyle{eptcs}

\end{document}